\documentclass[11pt]{article}
\usepackage{amsmath,amssymb,amsthm,amsfonts}
\usepackage{latexsym}
\usepackage{epsfig}
\usepackage[right=0.8in, top=1in, bottom=1in, left=0.8in]{geometry}
\usepackage{setspace}
\usepackage{bbold}
\usepackage{mathtools}
\usepackage{dirtytalk}
\usepackage{enumerate}
\usepackage{color,soul}
\usepackage{stackengine}
\usepackage{algpseudocode}
\usepackage{multicol}
\usepackage{extarrows}
\usepackage[linesnumbered,ruled,vlined]{algorithm2e}
\usepackage{verbatim}
\usepackage[backref,colorlinks,citecolor=blue,bookmarks=true]{hyperref}
\usepackage[nameinlink]{cleveref}
\usepackage[normalem]{ulem}

\usepackage{float}
\usepackage[section]{placeins}

\usepackage{scalerel,stackengine}
\stackMath
\newcommand\reallywidehat[1]{%
\savestack{\tmpbox}{\stretchto{%
  \scaleto{%
    \scalerel*[\widthof{\ensuremath{#1}}]{\kern-.6pt\bigwedge\kern-.6pt}%
    {\rule[-\textheight/2]{1ex}{\textheight}}
  }{\textheight}%
}{0.5ex}}%
\stackon[1pt]{#1}{\tmpbox}%
}

\newcommand{\zo}{\{0,1\}}

\newcommand{\N}{{\mathbb{N}}}

\newcommand{\R}{{\mathbb{R}}}

\newcommand{\eps}{\epsilon}
\newcommand{\tO}{\tilde{O}}

\newcommand{\polylog}{\mathop\mathrm{polylog}\nolimits}

\newcommand{\agl}[0]{{\cA_{\text{GL}}}}
\newcommand{\abl}[0]{{\cA_{\text{BL}}}}

\newtheorem{theorem}{Theorem}[section]
\newtheorem*{theorem*}{Theorem}

\newtheorem{corollary}[theorem]{Corollary}
\newtheorem{lemma}[theorem]{Lemma}

\newtheorem{definition}[theorem]{Definition}
\newtheorem{claim}[theorem]{Claim}
\newtheorem{fact}[theorem]{Fact}

\newtheorem*{oquestion*}{Open Question}

\makeatletter
\newtheorem*{rep@theorem}{\rep@title}
\newcommand{\newreptheorem}[2]{%
\newenvironment{rep#1}[1]{%
 \def\rep@title{#2 \ref{##1}}%
 \begin{rep@theorem}}%
 {\end{rep@theorem}}}
\makeatother
\newreptheorem{theorem}{Theorem}

\DeclareMathOperator{\E}{{\mathbb E}}

\DeclareMathOperator*{\argmax}{\mathop{\text{argmax}}}
\newcommand{\norm}[1]{\left\lVert#1\right\rVert}

\usepackage[pdftex,dvipsnames]{xcolor}  
\usepackage[utf8]{inputenc}
\usepackage{amsmath}
\crefname{ineq}{inequality}{inequalities}
\creflabelformat{ineq}{#2{\upshape(#1)}#3}
\usepackage{cryptocode}   
\graphicspath{ {img/} }
\usepackage[colorinlistoftodos,prependcaption,textsize=tiny]{todonotes}
\newcommandx{\unsure}[2][1=]{\todo[linecolor=red,backgroundcolor=red!25,bordercolor=red,#1]{#2}}
\newcommandx{\change}[2][1=]{\todo[linecolor=blue,backgroundcolor=blue!25,bordercolor=blue,#1]{#2}}
\newcommandx{\info}[2][1=]{\todo[linecolor=OliveGreen,backgroundcolor=OliveGreen!25,bordercolor=OliveGreen,#1]{#2}}
\newcommandx{\improvement}[2][1=]{\todo[linecolor=Plum,backgroundcolor=Plum!25,bordercolor=Plum,#1]{#2}}
\newcommandx{\thiswillnotshow}[2][1=]{\todo[disable,#1]{#2}}

\newcommand{\ns}[1]{{\color{red}$\ll$\textsf{#1 --Negev}$\gg$\marginpar{\tiny\textbf NS}}}
\newcommand{\csnote}[1]{{\color{purple}$\ll$\textsf{#1 --Cliff}$\gg$\marginpar{\tiny\textbf CS}}}
\newcommand{\aanote}[1]{{\color{blue}$\ll$\textsf{#1 --Alex}$\gg$\marginpar{\tiny\textbf AA}}}
\newcommand{\ssnote}[1]{{\color{violet}$\ll$\textsf{#1 --Sandip}$\gg$\marginpar{\tiny\textbf SS}}}

\renewcommand{\aanote}[1]{}
\renewcommand{\ns}[1]{}
\renewcommand{\csnote}[1]{}
\renewcommand{\ssnote}[1]{}

\newcommand{\remove}[1]{}

\DeclareMathOperator{\parent}{parent}

\def\poly{\operatorname{poly}}
\def\polylog{\operatorname{polylog}}

\renewcommand{\bf}{\normalfont \bfseries}

\renewcommand{\paragraph}[1]{\vspace{0.7ex}\noindent{\bf #1}}

\newcommand{\IC}{\mathcal{I}}

\newcommand{\LCS}[0]{\textsc{LCS}\xspace}
\newcommand{\LIS}[0]{\textsc{LIS}\xspace}

\newcommand{\OPT}[0]{\textsf{OPT}\xspace}

\DeclareMathOperator{\uniform}{\mathtt{Uniform}}

\newcommand{\1}[0]{\mathbb{1}}

\newcommand{\bX}{\mathbf{X}}
\newcommand{\bY}{\mathbf{Y}}
\newcommand{\bZ}{\mathbf{Z}}

\DeclareMathOperator{\Ber}{\mathtt{Ber}}
\DeclareMathOperator{\Bin}{\mathtt{Bin}}

\def\1{{\mathbf{1}}}

\newcommand{\cA}{\mathcal{A}}

\newcommand{\cD}{\mathcal{D}}
\newcommand{\cE}{\mathcal{E}}
\newcommand{\cF}{\mathcal{F}}
\newcommand{\cG}{\mathcal{G}}

\newcommand{\cI}{\mathcal{I}}
\newcommand{\cJ}{\mathcal{J}}

\newcommand{\cP}{\mathcal{P}}

\newcommand{\cR}{\mathcal{R}}
\newcommand{\cS}{\mathcal{S}}

\newcommand{\cY}{\mathcal{Y}}

\newcommand{\Pn}[0]{{\mathcal{P}}}

\def\colorful{1}

\ifnum\colorful=1

\newcommand{\blue}[1]{{{\color{blue}#1}}}

\fi
\ifnum\colorful=0

\newcommand{\blue}[1]{{{#1}}}

\fi

\newcommand{\ignore}[1]{{}}



\newcommand{\cor}[0]{1}
\newcommand{\uncor}[0]{0}
\newcommand{\reslis}[0]{\textsf{Block-LIS}\xspace}
\newcommand{\genlis}[0]{\textsf{Genuine-LIS}\xspace}
\newcommand{\blocklis}[0]{{\textsf{Block-LIS}}}

\newcommand{\prectree}[0]{Precision-Tree\xspace}
\newcommand{\prectrees}[0]{Precision-Tree\xspace}
\newcommand{\estreslis}[0]{\textsc{EstimateBlockLIS}\xspace}
\newcommand{\estgenlis}[0]{\textsc{EstimateGenuineLIS}\xspace}

\title{
Estimating the Longest Increasing Subsequence in Nearly Optimal Time
}

\author{
Alexandr Andoni\thanks{Research supported in part by NSF grants (CCF-1617955, CCF-1740833, CCF-2008733), and Simons Foundation (\#491119).}\\Columbia University\\\texttt{andoni@cs.columbia.edu} \and 
Negev Shekel Nosatzki\thanks{Research supported in part by NSF grants (CCF-1617955 and CCF-1740833), and Simons Foundation (\#491119).}\\Columbia University\\\texttt{ns3049@columbia.edu}
\and
Sandip Sinha\thanks{Supported by NSF grants  CCF-1714818, CCF-1822809, IIS-1838154, CCF-1617955, CCF-1740833, and by the Simons Collaboration on Algorithms and Geometry.}
\\
Columbia University\\
\texttt{sandip@cs.columbia.edu}
\and
Clifford Stein\thanks{Research partly supported by NSF Grants  CCF-1714818 and CCF-1822809.}
\\
Columbia University\\
\texttt{cliff@cs.columbia.edu}
}
\date{}

\begin{document} 

\maketitle

\begin{abstract}

Longest Increasing Subsequence (\LIS) is a fundamental statistic 
of a sequence, and has been studied for decades. While the \LIS of a sequence of length $n$ can be computed exactly in time $O(n\log n)$, the complexity of estimating the (length of the) \LIS in sublinear time, especially when \LIS $\ll n$, is still open.

We show that for any $n \in \N$ and $\lambda = o(1)$, there exists a (randomized) \emph{non-adaptive} algorithm that, given a sequence of length $n$ with \LIS $\ge \lambda n$, approximates the \LIS up to a factor of $1/\lambda^{o(1)}$ in $n^{o(1)} / \lambda$ 
time.
Our algorithm improves upon 
prior work substantially in terms of both approximation and run-time:
(i) we provide the first sub-polynomial approximation for $\LIS$ in sub-linear time; and (ii) our \emph{run-time complexity} essentially matches the trivial \emph{sample complexity} lower bound of $\Omega(1/\lambda)$, which is required to obtain any non-trivial approximation 
of the $\LIS$. 

As part of our solution, we develop two novel ideas which may be of independent interest. First, we define a new \genlis problem, in which each sequence element may be either genuine or corrupted. In this model, the user receives unrestricted access to the actual sequence, but does not know a priori which elements are genuine. 
The goal is to estimate the $\LIS$ using genuine elements only, with the minimal number of tests for genuineness. 
The second idea,
{\em Precision Tree}, enables accurate estimations for composition of general functions from ``coarse'' (sub-)estimates.
{\em Precision Tree} essentially generalizes classical precision sampling, which 
works only for summations.
As a central tool, the {\em Precision Tree} 
is 
pre-processed on a set of samples, which thereafter is repeatedly 
used by multiple components 
of the algorithm, improving their
amortized complexity. 
\end{abstract}















\ns{Address comments in Section 6, in particular: update language}


\ns{Address other small comments in section 7, Alex to consider moving Candidate intervals intuition to earlier on.}

\ns{finish changing $y \cap Y$ to $y(X,Y)$ throughout paper (in particular section 7 complexity analysis and invariants)}

\ns{formalize better proofs of section 8}


\ns{explain how to choose $\beta = \Theta(\log n)$ and $\tau$ with $\tau, n$ are powers of $\beta$}

\ns{make $\zeta$ a global parameter, since right now each w.h.p is $1-n^{-10}$, where $n$ is the instance size, would be nicer to have all w.h.p the same. For now I use $\beta$ for that.}






\newpage

{
  \hypersetup{linkcolor=black}
  \tableofcontents
}
\thispagestyle{empty}

\newpage

\setcounter{page}{1}

\section{Introduction} \label{intro}

Longest Increasing Subsequence (LIS) is a fundamental measure of a sequence, and has been studied for decades. Near linear-time algorithms have been known for a long time, for example,  the Patience Sorting algorithm \cite{Hammersley72,Mallows73} finds a LIS of a sequence of length $n$ in time $O(n\log n)$. The celebrated Ulam's problem asks for the length of a LIS in a random permutation; see discussion and results in \cite{AldousD99}.
 \LIS is also an important special case of the problem of finding a Longest Common Subsequence (LCS) between two strings, as \LIS is \LCS when one of the strings is monotonically increasing. Recently, there has been significant progress in approximation algorithms for LCS of two or more strings \cite{rubinstein2019approximation, DS21, GKLS21, BD21}. Moreover, \LIS is often a subroutine in LCS algorithms: for example, when strings are only mildly repetitive \cite[Chapter 12]{Gus-book}, or more recently in approximation algorithms \cite{HSSS19,Nosatzki2021ApproximatingTL}. Longest increasing subsequences have multiple applications in areas such as  random matrix theory, representation theory, and physics \cite{AldousD99}, and the related \LCS problem also has multiple applications in  bioinformatics, and is used for data comparisons such as in the \texttt{diff} command. 

In the quest for faster algorithms, researchers started studying whether
 we can estimate the length of a \LIS (denoted $\LIS$ as well) in {\em sublinear time}. An early version of this question underpins one of the first sublinear-time algorithms: to test whether an array is sorted, or monotonically increasing \cite{EKKRV00} --- i.e., whether the length of the \LIS is $n$ or is at most $(1-\eps)n$.  \cite{EKKRV00} gave a $O(\log n/\eps)$ 
 time  algorithm, and this running time was later shown to be  tight \cite{ACCL04, Fischer_comparisons_testing}. Since then, there have been numerous influential results on testing monotonicity and other similar properties; see, e.g., \cite{DodisGLRRS99, PRR06, Fischer01theart, CS17,BCS20,BS19ulam,PRW20, BCLW19, AN-ulam, SunWoodruff_LCS_LIS} and the book \cite{goldreich2017introduction}.

While monotonicity results focus on the case when $\LIS\approx n$, 
the case when 
$\LIS \ll n$ has seen much less progress.
The first result for this regime \cite{SaksS17} shows how to
$(1+\epsilon)$-approximate the length when the $\LIS$ is still large:
they distinguish the case when $\LIS \geq \lambda n$ from the case when $\LIS \leq (\lambda - \eps)n$ (for any $\lambda > 0$) 
using 
$(1/\eps)^{O(1/\eps)}\log^{O(1)}n$ time, which only gives a $(1 + \eps)$-factor approximation in truly sublinear time if $\lambda = \Omega(\log \log n / \log n)$. 
For arbitrary $\lambda<1$, assuming that $\LIS\ge \lambda n$ 
\cite{rubinstein2019approximation} gave an algorithm that achieves $O(1/\lambda^3)$-factor
approximation of $\LIS$ in $\tilde{O}(\sqrt{n}/\lambda^7)$ time.\footnote{The formal definition of an approximation algorithm here is one that 
has to return $\widehat{\LIS}$ such that $\widehat{\LIS} \leq \LIS$, and the approximation factor is $\LIS/\widehat{\LIS}$. 
Equivalently, one can conceptualize an $\alpha$-approximation algorithm (for $\alpha \geq 1$) as one that can distinguish the case when $\LIS \ge \lambda n$ from the case when $\LIS< \lambda n/\alpha$.
}
Recently, \cite{MS21} improved upon this result by presenting an algorithm with approximation $O(1/\lambda^\eps)$ 
and runtime $O(n^{1 - \Omega(\eps)} (\log n / \lambda)^{O(1/\eps)})$. Independently, 
\cite{NV20} 
obtained a non-adaptive $O(1/\lambda)$-approximation algorithm using 
$\tilde{O}(\sqrt{r}/\lambda^{2})$ queries, where $r$ is the number of distinct values. The authors of \cite{NV20} also proved a lower bound, showing that any \emph{non-adaptive} algorithm that estimates
the $\LIS$ to within an additive error of $\eps n$ 
requires $(\log n)^{\Omega(\log (1/\eps))}$ queries.

To put the above results into context, contrast \LIS  with the problem of estimating the weight of a binary vector of length $n$: when the weight is at least $\lambda n$, 
we can 
approximate the weight up to a factor of $1 + \eps$ by sampling $O_\eps(1/\lambda)$ positions, which is optimal. So far, one cannot rule out that a similar performance is achievable for estimating $\LIS$ too, in fact for $\lambda$ as large as $1/\log n$. 
The aforementioned prior results not only have approximation factors that are polynomial in $1/\lambda$ but also time / sample complexities that are worse by a polynomial factor in $n$ or $1/\lambda$. 
Hence the following guiding question remains open:

\begin{center}{\em
    Can we estimate the length of \LIS in essentially the time needed to estimate the weight of a binary vector?}
\end{center}

\subsection{Our contributions}
In this paper we come close to answering the above question in the affirmative, 
by obtaining an algorithm that runs in time 
near-linear in $1/\lambda$, 
and achieves 
sub-polynomial in $1/\lambda$ approximation. 
We show the following:

\begin{theorem}[Main theorem]\label{thm::lis_main}
	For $n \in \N$, and any $\lambda  = o(1)$, 
	there exists a (randomized) non-adaptive 
	algorithm that, given a sequence of length $n$ with $\LIS\ge \lambda n$, 
	approximates the length of $\LIS$ up to a
	$1/\lambda^{o(1)}$
	factor in $O\left(\tfrac{1}{\lambda}\cdot n^{o(1)}\right)$ time with high probability. 
\end{theorem}

We find this result quite surprising, as one may guess that 
when $\lambda\approx 1/\sqrt{n}$ (which is the \LIS length of a random permutation), one would need to read essentially the entire sequence (implying a $\Omega(1/\lambda^2)$ lower bound on the number of queries needed). 
On the contrary, when $\lambda \approx 1/\sqrt{n}$, our run-time (and hence sample complexity) nearly matches the {\em streaming complexity} from \cite{GJKK07,EJ08} (albeit with worse approximation).


\paragraph{Technical Contributions.} Our main technical contributions are three-fold (see technical overview in \Cref{technical_overview}):
\begin{itemize}
\item
We define (the \LIS problem in) a new sublinear computational model,
which we call $\genlis$ and which may be of independent interest. In the $\genlis$ problem, we are given a sequence $y\in \N^n$, with a caveat 
that only some of the elements of $y$ are ``genuine'' and the others are ``corrupted'', a property we can test an element for. The goal is to estimate the length of $\LIS$ {\em among the genuine} elements of $y$,
using as few tests as possible, while the values of $y$ are known to the algorithm ``for free". 
\item
We show how one can efficiently reduce the \LIS problem to \genlis and vice versa. These reductions together constitute the backbone of our recursive algorithm.
\item
To obtain the promised query complexity, we develop a new data structure, that we call a \prectree. This data structure samples non-adaptively (but non-uniformly) elements of the input $y$ and preprocesses them for efficient operations on the sampled elements. This part is the only part of the algorithm that samples the input string --- the aforementioned recursive calls of \LIS and \genlis problems access this data structure only. Furthermore, the data structure allows one to compute accurate estimations for composition of general functions from ``coarse'' (sub-)estimates --- a requirement for our recursive approach.
\end{itemize}

While this paper makes progress in understanding the complexity of estimating the $\LIS$, the following important question remains open:

\begin{oquestion*}[($1+\eps$)-approximation]
	Does there exist an
	algorithm that, given a sequence of length $n$ with $\LIS\ge \lambda n$, estimates the $\LIS$ up to a
	$1+\eps$
	factor in $O\left(\tfrac{1}{\lambda}\cdot n^{o(1)}\right)$ time with probability 2/3?
\end{oquestion*}

We believe that finding an improved approximation algorithm for the $\genlis$ problem above may lead to a $(1+\eps)$ approximation algorithm for \LIS.

We give a technical overview of our algorithm in  \Cref{technical_overview}, 
after setting up preliminaries 
in \Cref{prelims}. 
\Cref{sec:technical:main} contains the proof of the main theorem, assuming results proved in subsequent sections. In \Cref{sec::precision}, we develop the \prectrees data structure that is used to improve the sample and time complexity 
of the main algorithm. 
The guarantees of the two primary subroutines 
are proved in \Cref{sec::genlis} and \Cref{sec::reslis}.  
In \Cref{sec:genlis_ext}, we extend these algorithms in certain ways critical to our final application.

\subsection{Related work}

Computing the length of \LIS has also been studied in the streaming model, where settling its complexity is a major open problem \cite{P44}. In this setting, the main question is to determine 
the minimum space required to estimate the length of \LIS by an algorithm reading the sequence left-to-right. 
\cite{GJKK07,EJ08} gave  deterministic one-pass algorithms to $(1+\eps)$-approximate the \LIS using $O(\sqrt{n})$ space, and matching lower bounds against deterministic algorithms were given by \cite{EJ08, GG07}. These lower bounds are derived using deterministic communication complexity lower bounds and provably fail to extend to randomized algorithms \cite{Chakrabarti2012_LIS}. However, no randomized algorithm requiring $o(\sqrt{n})$ space is known for this problem either. We note that our algorithm can be used in the streaming setting, yielding streaming complexity $O(n^{1/2+\eps})$ for approximation $n^{o(1)}$.

There has been much more success on the ``complement'' problem of estimating the distance to monotonicity, i.e., $d_m:= n-\LIS$, in both the sublinear-time and the streaming settings. In the random-access setting, the problem was first studied in \cite{ACCL04}, and later in \cite{SaksS17}, who gave an algorithm that achieves $(1+\eps)$-approximation in time $\poly(1/d_m, \log n)$ for any constant $\eps>0$. These algorithms have been also used, indirectly, for faster algorithms for estimating Ulam distance \cite{AIK-product, NSS17} and smoothed edit distance \cite{AK-smoothed}. 

In the streaming setting, several results were obtained which achieve 
$O(1)$-approximation of $d_m$ using $\polylog(n)$ space \cite{GJKK07, EJ08}. This culminated in a randomized $(1+\eps)$-approximation algorithm using $\polylog(n)$ space by Saks and Seshadhri \cite{SS13_streaming_LIS}, and a deterministic $(1+\eps)$-approximation algorithm using $\polylog(n)$ space by Naumovitz and Saks \cite{NS15}. \cite{NS15} also showed space lower bounds against $(1+\eps)$-approximation streaming algorithms of $\Omega(\log^2 n/\eps)$ (deterministic) and $\tilde{\Omega}(\log^2 n/\eps)$ (randomized).

The \LIS problem has also been studied recently in other settings, such as the MPC and the fully dynamic settings. 
\cite{IMS17_MPC} gave a $(1+\eps)$-approximation, $O(1/\eps^2)$-round MPC algorithm for $\LIS$ whenever the space per machine is $n^{3/4+\Omega(1)}$. In the dynamic setting, a sequence of works \cite{MS20_dynamic_LIS, GJ21_dynamic_LIS} culminating in \cite{KS21}, gave the first exact dynamic \LIS algorithm with sublinear update time, and also gave a deterministic algorithm with update time $n^{o(1)}$ and approximation factor $1-o(1)$. \cite{GJ21_lower_bound_dynamic_LIS} showed conditional lower bounds on update time for certain variants of the dynamic $\LIS$ problem.

\paragraph{Subsequent Work.}
The main result from this paper was very recently used for estimating the Longest Common Subsequence (\LCS) problem. In particular, \cite{Nosatzki2021ApproximatingTL} gave the first sub-polynomial approximation for \LCS in linear time, invoking our algorithm as a sub-routine. 

\section{Preliminaries} \label{prelims}


	
\paragraph{Sequences and intervals.}
Given a set $\Omega$ and $n \in \N$, a sequence $y = (y_1, y_2, \cdots, y_n) \in \Omega^n$ is an ordered collection of elements in $\Omega$.
A block sequence $y \in \Omega^{n \times k}$ is a partially ordered collection of elements in $\Omega$.  Abusing notation, we will also allow for some elements in a block to be ``null'' --- for example, we write $y\in \N^{n\times k}$ to also mean $y\in \{\N,\bot\}^{n\times k}$, where $\bot$ is ``null''.

For $\ell \in [n]$, we say that $z = (z_1, \cdots, z_\ell)$ is a subsequence of $y$ of length $\ell$, and denote it by $\{y_{i_j}\}_j$, if there exist integers $1 \leq i_1 < i_2 < \cdots < i_\ell \leq n$ such that $z_j = y_{i_j}$ for all $j \in [\ell]$. We refer to the indices $i_j$ as coordinates and the values $y_{i_j}$ as element values. 
For a block sequence $y$, we say that $z = (z_1, \cdots, z_\ell)$ is a subsequence of $y$ of length $\ell$, and denote it by $\{y_{w_j}\}_j$, if there exist integers $1 \leq i_1 < i_2 < \cdots < i_\ell \leq n$ such that $z_j \in y_{i_j,*}$ for all $j \in [\ell]$. 


Define the interval space $\cI \triangleq \{[a,b] \mid a,b \in \N, a \leq b \} \cup \{[a,b) \mid a \in \N, b \in [1,\infty], a < b \}$. 
For $I \in \cI$, we use $|I|$ to denote 
$|I \cap \N|$, i.e., the number of natural numbers contained in interval $I$. We often 
refer to an $x$-interval $X \in \cI \cap 2^{[n]}$ as an interval of coordinates, and a $y$-interval $Y \in \cI$ as an interval of element values.

For a block sequence $y \in \N^{n \times k}$ and a $y$-interval $Y \in \cI$, we write $y \cap Y$ to denote the \emph{multi-set} of elements in $y$ that are also in $Y$. Also, for $X\subseteq[n]$, $y(X)$ is the sequence with first coordinates restricted to $X$. We also define $y(X,Y):=y(X) \cap Y$.



\medskip

\paragraph{Monotonicity.} 
We 
 define monotone sets as follows:
 
 \begin{definition}[Monotone sets]
		Fix a (potentially partial) ordered set $(\Omega, <)$. We say that a set 
	$P \subseteq \N \times \Omega$ is \emph{monotone} if 
	for all $(i,u),(j,v) \in P \times P$, we have (i) $i = j \Leftrightarrow u=v$; and (ii) $i < j \Leftrightarrow u < v$.
\end{definition}

%
%

Note that this definition captures the notion of an increasing subsequence. In particular, for the standard notion of an increasing subsequence over natural numbers, we take $\Omega = \N$ and $<$ as the usual ``less than" relation over $\N$ (a total order). However, we will need this more general definition to consider increasing subsequences over other partially ordered sets, like the space of intervals $\cI$.

For a finite set $P \subset \N \times \Omega$, a longest increasing subsequence ($\LIS$) of $P$ is a monotone set $Q \subset P$ of maximum cardinality. 
We often use $\OPT$ to refer to a particular $\LIS$, and use $|\OPT|$ to denote its length.

\medskip

\paragraph{Distributions.} 
For $p \in [0,1]$, we use $\Ber(p)$ to denote the Bernoulli distribution with success probability $p$, and $\Bin(n,p)$ to denote the binomial distribution with parameters $n$ and $p$. 
By convention, we project $p$ to the range $[0,1]$ whenever $p > 1$ or $p < 0$.

We use ``i.i.d. random variables" to mean that a collection of random variables is independent and identically distributed, and use ``sub-sampling i.i.d. with probability $p$" 
to mean that each element is sampled independently with equal probability $p$.

\medskip

\paragraph{Operations on Vectors, Sets, Functions.}
For a set $A \subset \R$ and a number $\alpha$, we define $A + \alpha := \{a + \alpha : a \in A\}$ and $\alpha A := \{\alpha \,a  : a \in A\}$.
We write 
$\log$ to denote binary logarithm.

The notation $\circ$ is used for function composition (i.e., $g \circ f(x) = g(f(x))$, and 
$\oplus$ is used for direct sum.
We use the notation $*$ as argument of a function, by which we
mean a vector of all possible entries. For example, $f(*)$ is a vector of
$f(i)$ for $i$ ranging over the domain of $f$ (usually clear from the
context).

We use $E_b(k)$ for $b,k \in \N \cup \{\infty\}$ to denote the set of powers of $b$ bounded by $k$, 
i.e.,
$E_b(k) \triangleq \{b^i \mid i \in \N\} \cap [1,k]$. We also define $E_b \triangleq E_b(\infty)$.
We use $\R_+$ to denote the set of non-negative real numbers.

\medskip

\paragraph{Approximations.}
Our constructions generate approximations which contain both multiplicative and additive terms. 
We use the following definition:
\begin{definition}[$(\alpha,\beta)$-Approximation]
	For $\alpha \geq 1$ 
	and $\beta>0$, an $(\alpha,\beta)$-approximation $\hat{q}$ of a quantity $q$ is an $\alpha$-multiplicative and $\beta$ additive estimation of $q$, i.e., $\hat{q} \in [q/\alpha - \beta, q]$.
\end{definition}

The following fact shows that 
to obtain a multiplicative approximation 
which is a function of the $\LIS$, it suffices to show $(\alpha,\beta)$-approximation.

\begin{fact}\label{ft::approx}
    Suppose we have an algorithm $\cA$ that for any $n \in \N, \lambda \in (0,1)$, for some 
    $q \in (0, 1)$, 
    outputs a $(1/\lambda^q,\lambda n)$-approximation for some unknown quantity $\ell \in [0,n]$ in time $t$. Then there exists an algorithm 
    that outputs 
    $\widehat{\ell} \in [\Omega(\delta^p),1] \cdot \ell$ in time $t + O(1)$, where $p \triangleq \tfrac{q}{1-q}$, as long as $\ell \geq  2\delta n$.
\end{fact}

\begin{proof}

The algorithm calls $\cA$ with parameter $\lambda = \delta^{1 + p}$, and returns the same output. Then, the upper bound is immediate by definition of $(\alpha,\beta)$-approximation. For the lower bound, as long as $\ell \geq 2\delta n$, we have that:
\[
\widehat{\ell} \geq \delta^{q(1+p)} \cdot \ell
- \delta^{1 + p} \cdot n
\geq
\delta^{p} \cdot \ell - \tfrac{1}{2}\delta^{p} \cdot \ell
\geq \tfrac{1}{2}\delta^{p}\ell \tag{$q = p/(1+p)$}
\]
\end{proof}

\paragraph{Other notation.} Notation $\tO(\cdot)$ hides $\polylog(n)$ factors, while $O^*(\cdot)$ hides a factor of $n^{o(1)}$. 

\section{Technical overview} 
\label{technical_overview}


Our starting point is the algorithm of \cite{rubinstein2019approximation} which achieves $O(1 / \lambda^3)$ approximation in time $\tO(\sqrt{n} / \lambda^7)$.
Let $\OPT$ be (the coordinates of) an optimum solution (LIS)
with length $|\OPT| \geq \lambda n$. Their algorithm consists of 
2 main
steps:

\ignore{
\begin{enumerate}
    \item 
    The input sequence $y$ 
    is divided into $\sqrt{n}$ contiguous blocks $y^{(1)}, \cdots, y^{(\sqrt{n})}$ of length $\sqrt{n}$ each. For each block $y^{(i)}$, it ``guesses" the interval $[s_i, \ell_i]$, where 
    $s_i$ and $\ell_i$ are the minimum and maximum values in $y^{(i)} \cap \OPT$ 
    respectively. This is done by sampling $O(1/\lambda)$ elements in $y^{(i)}$, and considering the $O(1/\lambda^2)$ intervals bounded by all possible pairs of sampled elements, called \emph{candidate intervals}.
    \item It generates a set of mutually disjoint \emph{pseudo-solutions}, which are sequences of $\Omega(\sqrt{n} \lambda)$ candidate intervals that are monotone, i.e., all values in a candidate interval in block $i$ are less than all 
    values in a candidate interval in block $j$ for all $i < j$. It estimates the quality of each pseudo-solution (the sum of $\LIS$ of the candidate intervals in it) using sub-sampling, and 
    outputs the largest quality.
\end{enumerate}
}

\begin{enumerate}
    \item 
    Partition $[n]$ into $\sqrt{n}$ disjoint contiguous x-intervals $X_1,X_2,\ldots,X_{\sqrt{n}}$ of length $\sqrt{n}$ each. 
    For $i \in [\sqrt{n}]$, let $y(X_i)$ 
    be the restriction of the input sequence $y$ to $x$-interval $X_i$. 
    The goal is to approximate the $y$-interval $Y_i \triangleq [s_i, \ell_i]$, where $s_i$ and $\ell_i$ are, respectively, the minimum and maximum values in $y(X_i \cap \OPT)$. To do this, 
    one samples $O(1/\lambda)$ elements in $X_i$, 
    and generates $O(1/\lambda^2)$ \emph{candidate $y$-intervals} (using pairs of sampled element values).
     
    \item Generate a set of mutually disjoint \emph{pseudo-solutions}, each of which is a sequence of $\Omega(\lambda \sqrt{n})$ candidate $y$-intervals that are monotone, i.e., all values in a candidate interval in x-interval $X_i$ are less than all 
    values in a candidate interval in x-interval $X_j$ for all $i < j$. Estimate the quality of each pseudo-solution (the sum of $\LIS$ of the candidate intervals in it) by simply 
    sampling $O(\log^{O(1)}n/\lambda^4)$ x-intervals 
    and computing the \LIS within relevant candidate intervals. Output 
    the largest quality.
\end{enumerate}

Their analysis proceeds as follows. For each x-interval $X_i$, 
some candidate $y$-interval is a good approximation of the interval $Y_i$ 
w.h.p., so the union of the \LIS in all pseudo-solutions covers a large fraction of $\OPT$. Since there are $O(1/\lambda^3)$ pseudo-solutions, the output is a $O(1/\lambda^3)$-approximation 
of $\OPT$. The runtime of $\tilde{O}(\sqrt{n}/\lambda^7)$ is dominated by the time required to evaluate the quality of pseudo-solutions to sufficient accuracy. One can also apply this technique 
recursively to improve the runtime, at the cost of worse approximation.

In \cite{MS21}, this result is improved by giving an algorithm for $\LIS$ with approximation factor $O(1/\lambda^\eps)$ 
and runtime $O(n^{1 - \Omega(\eps)} \cdot (\log n / \lambda)^{1 / \eps})$ for any constant $\eps > 0$. 
The polynomial dependence on $1/\lambda$ seems essential since the algorithm in \cite{rubinstein2019approximation} is used as a sub-routine. 
In \cite{NV20}, a \emph{non-adaptive} algorithm is presented with sample complexity $\tilde{O}(\sqrt{r}/\lambda^2)$ and approximation $O(1/\lambda)$, 
where $r$ is the number of distinct values. 
Again, the polynomial dependence of the sample complexity on $r$ and $1/\lambda$ seems intrinsic.

There are several fundamental obstacles in improving these bounds, and in particular, getting the runtime down all the way to $1/\lambda$ while improving the approximation. One such obstacle is that the straight-forward approach (as in prior work) of independent recursions cannot obtain better than $1/\lambda^2$ run-time in the worst-case. 
The issue is that even if one is able to isolate all possible $y$-interval ranges efficiently, 
there must still be at least $1/\lambda$ such ranges (optimally), 
and recusing of each such $y$-interval 
without knowing a priori which coordinates contain the relevant $y$-values 
would require $1/\lambda$ samples just to find a single sample of relevance in each (sub-)instance (i.e. for each $y$-interval). We will return to this obstacle later and show how we improve the {\em amortized complexity} for it.

\subsection{Our approach} \label{tech_overview_our_approach}
Similarly to \cite{rubinstein2019approximation}, we generate
candidate $y$-intervals based on random samples, and generally look for increasing sequences of intervals, each with a large local $\LIS$.
Beyond this general similarity, our algorithm develops a few new ideas, in particular in how we work with these candidate intervals, and especially how we find \LIS's among these candidate intervals.
In particular, our contribution can be seen through three main components.

{\bf First}, we introduce a new problem, termed \genlis, which captures the problem of estimating the longest (sub-)sequence of increasing intervals, each with a large ``local \LIS''. This problem is a new model for sublinear-time algorithms, which has not appeared in prior literature to the best of our knowledge.

The \genlis problem
is defined as follows. 
The input is a sequence $g$ where each element 
is associated with 
an additional flag signifying whether it is \emph{genuine} or not, i.e., $g \in (\N \times \{\uncor,\cor\})^{n}$. 
One receives unrestricted access to the elements (i.e., the first coordinates) which is denoted by $g(*)_1 \in \N^{n}$. However, one must ``pay'' to test whether an element is genuine; 
this is determined by the second coordinates $g(*)_2 \in \{0,1\}^{n}$, referred to as genuineness flags. The goal is to compute the length of the longest increasing subsequence of $g(*)_1$ 
restricted to genuine elements, i.e., of $g((g(*)_2)^{-1}(1))_1$, 
using as few 
{\em tests for genuineness} as possible. 

{\bf Second}, we show how to efficiently reduce the \LIS problem to a (smaller) 
instance of the \genlis problem, where the ``genuineness'' flag of an element corresponds to \LIS of an $x$-interval being large enough. In particular, this is where we generate the aforementioned candidate $y$-invervals. 
The algorithm reduces the problem to finding a (global) \LIS amongst a set of candidate intervals, each with a large (local) \LIS. We solve this via a composition of 2 sub-problems: the $\genlis$ problem takes care of the ``global \LIS'' of candidate-intervals, while the standard \LIS problem recursively estimates the ``local \LIS'' (these ideas and terms will be made precise later). The main challenge here is to generate candidate intervals in a {\em succinct} manner: we manage to reduce the number of sampled anchor needed to $\approx 1/\lambda$, as well as limiting the number of candidate intervals to be {\em near-linear} in the number of sampled points. 

{\bf Third}, to optimize the sampling of the string, we use a novel sampling scheme endowed with a data structure, termed {\em \prectree}. At preprocessing, the structure samples a number of positions in the input string (non-adaptively), which are not uniform but rather correspond to a tree with non-uniform leaf levels (hence the name). We then build a global data structure on these samples for efficient access. 
In particular, we develop a {\em tree decomposition procedure} (to be discussed later), and repeatedly use 
{\em sub-trees} across different sub-instances we generate, improving the overall sample complexity. 
We highlight the fact that we {\em reuse} both the samples (and hence randomness) across different (sub-)instances for both of the \LIS and \genlis problems.
In addition, the \prectree data structure provides a quick way to narrow down on all samples in some $x$-interval whose value is in given $y$-interval $Y$.

Before continuing this overview, we formally define a 
slightly generalized version of the $\LIS$ problem, termed $\reslis$, as well as the aforementioned \genlis problem. Following that, we 
provide an overview of the algorithms to solve these two problems. The high-level flow of our algorithm is described in \Cref{fig::main_flow}.


\begin{figure}[h!]
\centering
\includegraphics[width=0.98\textwidth]{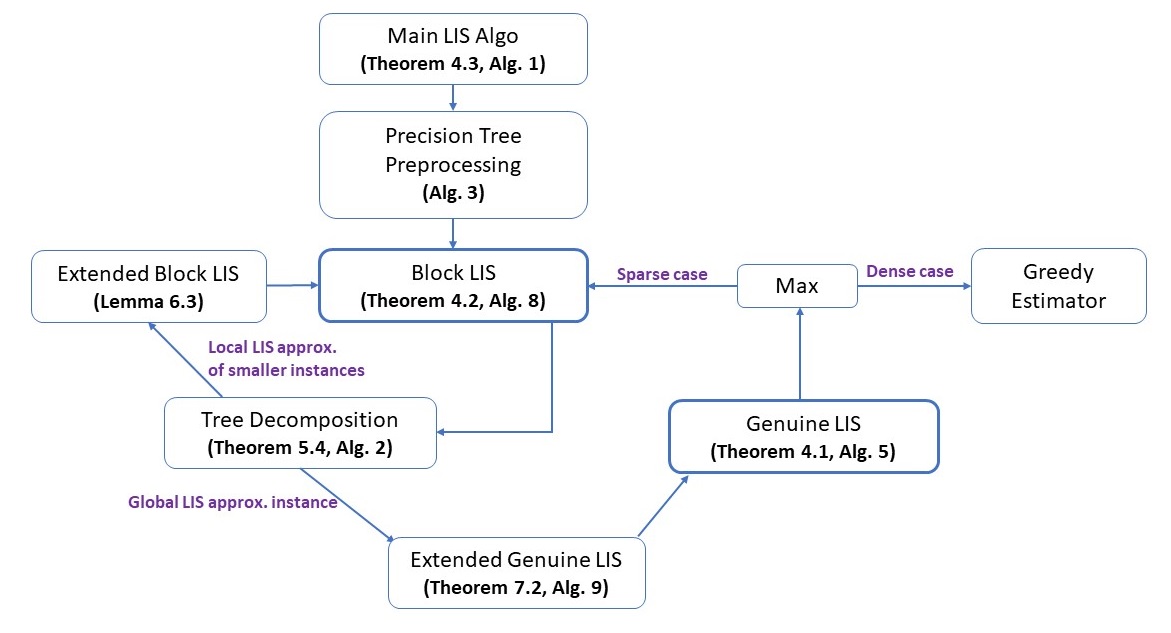}
\vspace{-4mm}
\caption{High level flow of the main algorithm (\Cref{thm::lis_main}). Each arrow represents one or more 
calls to the corresponding algorithm.}
\label{fig::main_flow}
\end{figure}


\paragraph{The \reslis problem.}
This problem extends standard \LIS problem in two ways.
First, the main input is a sequence of $n$ blocks, each containing (at most) $k$ elements,  
and 
at most one element in each block 
may participate in a sequence. Second, we are also given a range of values $Y \in \cI$, such that each element of a subsequence must be in $Y$.

Formally, we define the $\reslis$ as follows:
given a sequence 
$y \in \N^{n \times k}$ and a range of values 
$Y \in \cI$, compute the length of a maximal sub-sequence $\OPT \subseteq [n] \times [k]$,
such that $\{(w_1,y_w)\}_{w \in \OPT}$ is monotone and $\{y_w\}_{w \in \OPT} \subseteq Y$. 
Sometimes, we also restrict the set of blocks to an interval $X \subseteq [n]$, and define the quantity $\blocklis(y,X,Y)$ as the longest 
increasing subsequence of $y \in \N^{X \times [k]}$ using elements in $Y$.

The standard 
$\LIS$ problem can be seen as the \reslis problem with 
$k=1$ and $Y = \N$ --- which is how we instantiate the original input sequence $y \in \N^n$ in our $\reslis$ algorithm. The main advantage of this
generalization appears when the sequence is {\em sparse in $Y$}, i.e. most elements cannot participate in any \LIS. Then, we show 
the multiplicative approximation factor for $\reslis$ is not only a function of the additive error $\lambda n$, 
but also of the total number of elements of $y$ 
in range $Y$ (i.e., $|y(X, Y)|$). In particular, we show that this 
approximation is a function of $\tfrac{\lambda n}{|y(X, Y)|}$.

While the generalization to blocks is not a core necessity of the algorithm (in fact, one can consider the values in each block in descending manner, yielding an equivalent problem with no blocks required), its main use comes from our need to instantiate several overlapping instances using the same \prectree data structure (to be discussed later).

Finally, we note that we similarly generalize the \genlis to work over blocks: input $g\in  (\N \times \{\uncor,\cor\})^{n\times k}$ and the \LIS is allowed to use only one (genuine) element per block.

\subsection{Main Algorithm}

The main algorithm for estimating the $\LIS$ is merely the following (see \Cref{alg::ALG_lis_main}):
\begin{enumerate}
	\item Preprocess \prectree $T=T_\delta(y)$ 
	of the input sequence with parameter $\delta \triangleq \lambda / n^{o(1)}$ (the exact $o(1)$ parameter will be made precise later).
	\item Run the \reslis algorithm that has access only to the data structure $T$. (we note that this algorithm itself recursively runs the algorithms for \genlis and \reslis.)
\end{enumerate}

\subsection{\genlis problem: overview of the algorithm \estgenlis}


Recall that in the \genlis problem with input $g \in (\N \times \zo)^n$, one receives unrestricted access to the elements (i.e., the first coordinates), but must ``pay'' to test whether an element is genuine (determined by the second coordinates). The goal is to compute the length of the longest increasing subsequence of $g(*)_1$ restricted to genuine elements, while minimizing the number of genuiness tests.

We derive two complementary solutions to the \genlis problem: one direct (without further recursion), and the other by reducing the problem to  standard sublinear-time $\LIS$ estimation. To highlight our novel contribution, we contrast it with the framework of  \cite{rubinstein2019approximation}, that implicitly gives a solution to the $\genlis$ problem, albeit with sub-optimal approximation. In particular, the following algorithm for \genlis is analogous to the algorithm in \cite{rubinstein2019approximation} based on generating pseudo-solutions (maximal sequences of candidate intervals).
\begin{enumerate}
    \item
    Using the first coordinates only, look for a  maximal {\em increasing subsequence} $P_1$ containing $\gtrapprox \lambda n$ elements and remove $P_1$ from the sequence.
    \item Repeat step (1) until there are no more sequences of length $\gtrapprox \lambda n$. This generates (disjoint) solutions $P_1,\ldots,P_t$ for some $t \lessapprox 1/\lambda$.
    \item Sub-sample $\approx 1 / \lambda$
    elements 
    from the union $P = \cup_{i \in t} P_i$ and check if each one is genuine. 
    Let $\kappa$ be the number of Genuine elements in this sampled set.
    \item Output $\lambda^2 n\cdot \kappa$.
\end{enumerate}

It is easy to see that this algorithm yields, with constant probability, a $(2/\lambda, \lambda^2 n)$-approximation with approximately 
$1/\lambda$ 
tests for genuineness. First, if the $\genlis$ 
is at most $2\lambda n$, then the output (being greater than $0$) is already a $(2/\lambda, \lambda^2 n)$-approximation, as 
\[
\tfrac{\genlis}{2/\lambda} - \lambda^2 n \leq \tfrac{2\lambda n}{2/\lambda} - \lambda^2 n = 0 \leq \genlis.
\]

Otherwise, 
$P$ contains most of the \genlis 
elements,
and since $t \lessapprox 1/\lambda$, 
the 
$\genlis$ 
is at least $\lambda/2$ 
times the number of genuine elements in $P$.

Thus, by the sampling mechanism above, we can estimate the proportion of genuine elements to sufficient accuracy with constant probability.
Observe that the approximation is in fact proportional to $1/\lambda \kappa$, and that the worst approximation happens when $P$ is ``sparse'', i.e.,  
when only approximately 
$\lambda n$ elements in $P$ are genuine. For such a {\em sparse case}, however, one can instantiate the $\genlis$ algorithm as a $\reslis$ one, restricted to the genuine elements. In fact, this ``sparse'' case is precisely where the approximation factor is minimal, i.e., $\tfrac{\lambda n}{|y|} =1$ where $y$ is the genuine-elements-only subsequence. Our improved approximation stems, in part,  from balancing between the dense and sparse cases as above.

\paragraph{Speeding up $\LIS$ extraction using dynamic $\LIS$.}
The $\genlis$ instances we generate are of size $\lessapprox 1/\lambda$, and our goal is to obtain overall run-time that is near-linear in $1/\lambda$
as well.
The standard, dynamic-programming solution for finding and extracting the optimal $\LIS$ each time for step 
(1) above can potentially incur an overhead that is quadratic in the instance size (each $\LIS$ extraction would take linear time, and we need to greedily extract up to a near-linear number of solutions), leading to a polynomially-worse time.
To ensure near-linear time, we adopt a fast dynamic $\LIS$ data structure from \cite{GJ21_dynamic_LIS}, and use it to iteratively extract near-maximal pseudo-solutions. 

\paragraph{Extensions to the \genlis algorithm.} 
Our complete algorithm needs three further extensions to the algorithm for the \genlis problem, as follows.

\begin{itemize}
    \item {\em Block version:} just like for \reslis, we define the block version of \genlis, where the input is a sequence $g\in (\N\times \{0,1\})^{n\times k}$, and we need to find the LIS among the genuine elements, while using at most one element per size-$k$ block.
    \item {\em Sparse (unbalanced) instances:} some encountered \genlis instances are sparse, with many ``null" elements, i.e., some blocks have less than $k$ elements (note that a ``null'' element is different from a ``non-genuine'' one, in that it does not need to be tested). When the instance $g$ consists of $m \ll nk$ non-null elements, we would like the approximation and runtime bounds to be a function of $m/n$ (average block size) instead of $k$ (maximum block size). To obtain the improved bound, we partition the blocks based on an exponential discretization of the number of non-null elements, and output the maximum over all instances, where each instance consists only of blocks containing a similar number
    of non-null elements.
    \item {\em \genlis over intervals:} the first coordinates of the \genlis instances we generate apriori consist of intervals rather than integers. The challenge is that we then need to find \LIS over the {\em partial order} of intervals $\cI$. To illustrate why this poses an additional challenge, note that the dynamic $\LIS$ data-structures from \cite{MS20_dynamic_LIS, GJ21_dynamic_LIS,KS21} (used to speed up our algorithm) cannot immediately handle partial order sequences. An ideal 
solution here would be a mapping $\varphi : \IC \rightarrow \N$ that approximately preserves the overall $\LIS$ over all subsequences, and therefore also preserves the overall $\LIS$ over the genuine elements. While we are unable to show a single mapping that works for all intervals, we partition the space of intervals into $\log (k/\lambda)$ sets $\cI_\ell$ --- intervals in $\cI$ whose length is in $[\ell, 2\ell)$ --- 
based on an exponential discretization of the interval length $\ell$, 
and provide a 
mapping $\varphi_\ell : \cI_\ell \rightarrow \N$ for each set of intervals $\cI_\ell$. 
We eventually output the maximal $\genlis$ result over all such 
maps. 
This costs us merely another $\log (k/\lambda)$-factor approximation, and a small additive error.
\end{itemize}



The formal statement for the algorithm to solve $\genlis$, named \estgenlis, is presented in \Cref{sec:technical:main}, and its 
description and analysis 
are in \Cref{sec::genlis}. The pseudo-code is presented in \Cref{fig:ALG_genlis}. The \genlis extensions are presented and analyzed in 
\Cref{sec:genlis_ext}.

\subsection{\reslis problem: overview of the algorithm \estreslis}
For the problem $\reslis(y,X,Y)$, we develop the algorithm
\estreslis. Fundamentally, we reduce the instance to a \genlis instance, where each 
genuineness flag is set to 1 if and only if the \LIS of a certain (smaller) sequence (which is itself a \reslis problem) is large enough.

Our reduction starts by partitioning $X$ into consecutive intervals $X_1,\ldots,X_\tau$ of equal length, 
where $\tau$ is a dynamic, carefully chosen branching factor, and is usually 
sub-polynomial in the instance size (here think $\tau=n^\eps$). Next, we sample approximately 
$\tau/\lambda$ blocks $w_j\in X$ (termed \emph{anchors}), and, for each $i\in[\tau]$, we collect all the $y$-values across all the blocks $w_j\in X_i$ that are also in $Y$, and generate a set of $y$-values $\cS_i \subseteq Y$.

Using each 
     set of $y$-values $\cS_i$, we construct \emph{candidate intervals} $\cY_i$. 
     While this part is similar to the construction from~\cite{rubinstein2019approximation}, we note that we need a more efficient construction to obtain the near-optimal sample complexity. In particular, we require the number of candidate intervals to be near-linear in $|\cS_i|$, and hence we 
     construct a {\em small} candidate interval set $|\cY_{i}| \approx |\cS_i|$, while still ensuring that this set covers all ``relevant options'' with only an extra logarithmic factor loss in approximation. In particular, instead of looking for all possible pairs of endpoints, 
     we choose the endpoints in dyadic fashion (in particular, their distance in the set $\cS$ must be a power of 2).
     
     Finally, we reduce the $\reslis$ problem to a $\genlis$ instance over $\tau$ blocks, where each block $i$ contains all the candidates intevals $\cY_i$ --- this \genlis instance thus captures ``global'' \LIS (over the candidate intervals). The first coordinates of the $\genlis$ instance are the candidate intervals $\cY_i$ themselves, while the second coordinates (i.e., the genuineness flags) indicate whether 
     the corresponding ``local'' $\blocklis(y,X_i,Y')$ estimated values 
     are above a certain threshold $\kappa$, for each $Y' \in \cY_i$.

     This threshold $\kappa$ itself depends on a parameter $\rho>1$, which characterizes the relation between the ``global" $\genlis$ instance and the ``local" \reslis instances. Intuitively, $\rho$ is such that (roughly)
     $1/\rho$ fraction of the $X_i$ intervals each have a $\lambda \rho$ fraction of them participating in the 
     \LIS. Consider two extreme cases, one where the $\LIS$ is uniformly distributed among all $X_i$ ($\rho=1$), and one where the $\LIS$ is maximally concentrated among a small subset of the $x$-intervals ($\rho=1/\lambda$). Then, it is more difficult to certify an increasing subsequence in the latter case, when the $\LIS$ is sparse. But in this case testing the genuineness of a local \LIS (inside a block $X_i$) is much easier, since we only need to establish that it is $\Omega(|X_i|)$. In general, there is a precise trade-off between the complexity of certifying the global $\LIS$ versus the complexity of certifying the local $\reslis$. A priori, we do not know $\rho$, so
     we simply 
     iterate over
     all possible magnitudes (again, by exponential discretization).
     
     Once we formulate the overall \reslis problem as a composition of a ``global'' $\genlis$ over multiple ``local'' $\reslis$ instances, 
     we decompose the problem through a procedure called {\em \prectree decomposition} which will be described next.

The formal statement for the \estreslis algorithm is presented in \Cref{sec:technical:main}, and the algorithm description and analysis are in 
\Cref{sec::reslis}. The pseudo-code (including sub-routines) are presented in \Cref{fig::ALG_reslis_sample,fig::ALG_reslis_decompose,fig::ALG_reslis_amended}. 

\subsection{\prectree data structure}

We note that a direct instantiation of the above algorithm yields the right approximation, but will require at least $\approx 1/\lambda^2$ queries into the input string $y$ (and hence run-time), and, moreover, require the algorithm to be {\em adaptive}. The source of this inefficiency is precisely the ``small slopes'' obstacle discussed earlier, which would lead to ``wasting samples'' in any na\"ive rejection sampling mechanism.
To improve the complexity, and to allow our sampling algorithm to be non-adaptive, we introduce the {\em \prectrees} data structure (\Cref{sec::precision}). While based on the {\em precision sampling} tool introduced in \cite{AKO-edit,AKO-precision}, the main development here is that we augment the resulting tree for special operations described henceforth.


Intuitively, the \prectree data structure can be thought of as a global data structure holding enough information for simulating random samples for multiple non-uniform (sub-)instances, specifically ones arising from recursion.
In particular, the data structure is an (incomplete) tree, where (lowest-level) leafs correspond to samples of some input $e \in \cE^n$. 
The initial tree is defined over the input string (i.e., $\cE=\N$),
but we also consider Precision-Trees over other domains $\cE^n$ as  described later (which are simply a recasting the original single tree). The initial tree for the string $y \in \N^n$ is constructed at random, using precision sampling, and this is the only time when we access the string $y$ (in the entire \LIS algorithm). 
Overall, the \prectree
data-structure (the initial one as well as the other trees we generate) supports several important properties we need:
\begin{itemize}
    \item 
Since our algorithm is recursive (composition of $\genlis$ and $\reslis$ algorithms), we need to provide samples for potentially very small contiguous subsequences of $y$ (depending on what was sampled higher up in the recursion). In other words, we need to be able to ``zoom into'' different location of the input string with the right precision. 

Our precision tree is sampled such that any sub-tree under any ``sampled'' node, is another \prectree instantiation (with potentially different ``precision'' parameter).

 \item
In contrast to the original precision sampling tool that was designed for the simple addition function (over elements), our 
application at hand requires more complicated functions (indeed, the
sub-linear \reslis and \genlis algorithms).

We show how we can {\em decompose} the tree in a van Emde Boas-like layout: by cutting it to obtain a number of different ``lower sub-trees'', together with a ``top tree'' whose domain are the lower sub-trees (as mentioned, all lower/top sub-trees are themselves a precision tree). Then we can run algorithms on each sub-tree, as well as a final algorithm recomposing the results from the sub-trees, using the top tree.

\item
Another crucial challenge is the following: as our overall algorithm recurses, we may be ``zooming in'' on the same part of $y$ multiple times, but focusing on different $y$-range intervals $Y$. Na\"ively, one would use rejection sampling here (rejecting samples with value not in $Y$), which however would yield a polynomially worse sample complexity. In addition, our recursion happens not only down the tree (i.e., on the lower sub-trees) but also on the top portion of the tree (top tree above).

We augment our precision tree to allow us to {\em reuse} the randomness over independent recursive calls into the same $x$-intervals, each focusing on a different $y$-value ranges. In particular, for such an $x$-interval we will have processed all the sampled entries to directly report samples with values in a desired $y$-interval $Y$ (note that this is essentially a 2D range reporting). In addition, the randomness is reused when computing the global function on the top portion of the tree as described above.


\end{itemize}


\subsubsection{Construction of the \prectree data structure} 
\label{sec:constPT}
We now describe the formal construction of \prectree.
Original precision sampling from \cite{AKO-edit, AKO-precision} is designed to estimate a {\em summation} function $a = \sum_{i=1}^n a_i$, for unknown $a_i\in[0,1]$, from ``coarse''
estimates for $a_i$.
In our application, we need to {\em generalize} precision sampling from a simple addition to {\em general functions}, allowing one to approximate $g \circ f$ where $g$ is a general function over $m$ coordinates and $f=(f_1,\ldots,f_n)$ consists
of $n$ {\em independent} functions on different parts of the input, sharing the same co-domain.

We define \prectree with parameter $\delta<1$, for a given vector of elements $e = (e_1,\ldots,e_n) \in \cE^n$, denoted $T_{\delta}(e)$, as follows. The tree $T$ is a trimmed version of the full $\beta$-ary tree of $\ell=\log_\beta n + 1$ levels, with $n=\beta^\ell$ leafs, where parameter $\beta>1$ is fixed; in particular, all nodes of $T$ have fan-out of either $\beta$ or 0.  Each leaf at level $\ell$
\footnote{By convention, the level of the root is 0.} is associated with an integer representing its location in the input string.
Each internal node $v$ is associated with an interval representing the unions of its leaf descendants (in the full tree); for example,
the root of $T$ is associated with the entire interval $[1,n]$, and its children with $[1,n/\beta]$, $[n/\beta+1,2n/\beta]$, etc). For a node $v$ in the tree, we use $\cE_v$ or $\cE(v)$ 
to denote the set of elements under $v$.



Given $n, \lambda$, and $\beta$, we construct the trimmed  \prectree $T_{\lambda}(e)$ as follows.  
We assign a (random) precision score $\Pn_v$ to each node $v$ by the following recursive procedure.
 We set 
 $\Pn_{\text{root}(T)} = \delta$.   Recursively we define the precision score of a node $v$ by 
\begin{equation}
    \Pn_v = \bZ_v \cdot \Pn_{\parent(v)} \text{ where } \bZ_v \sim_{i.i.d.} \uniform(\{1,2,3,\ldots,\beta/4\})
    \ .
    \label{eq:precision-tree}
\end{equation}

Most importantly, for a node $v$, we recurse into its children only if $\Pn_v \le 1$; otherwise we stop the recursion for $v$ with $\Pn_v>1$ ($v$ has 0 children). If $v$ is a leaf (at level $\ell$), with the associated integer $i\in [n]$, then we sample the corresponding element $e_i$ of the input string $e$. One should think of the ``precision score'' $\Pn_v$ of some internal node $v$ as the degree to which the sub-tree rooted at $v$ is sampled: the lower $\Pn_v$, the higher precision and the amount at which the sub-tree is sampled. One can prove that the sample complexity (into $y$) is bounded by $1/\delta\cdot\log^{O(\log_\beta n)}(\beta)$ (see \Cref{lm::sample_bound}).

At the beginning, we generate one such tree $T_{1/c}(y)$ 
for the input string $y \in \N^n$, where $\delta=1/c$ for some $c = c(n,\lambda)\ge1$ that we refer to as the ``precision tree complexity''.  After obtaining these samples, the rest of the algorithm has no further access to $y$, and hence the sampling is non-adaptive.
Most importantly, using this single tree $T_{1/c}(y)$, we will simulate \prectree access for different trees over different inputs and parameters in our algorithms using a certain procedure which we call the Trim-Tree algorithm (see \Cref{fig:trimtree}).

We preprocess the sampled positions so that we can quickly retrieve all samples in some $x$-interval $X \subseteq [n]$ 
and some $y$-interval $Y$ through an ancillary data-structure which we call SODS.
The SODS is used for a ``rejection sampling''-type operations where we need a number of samples from a substring of input $y$ (or later on for some intervals of blocks), but care only about values in a certain range $Y$. 
Below, we abuse notation to use $T$ to denote both the conceptual \prectree and the data structure implementing it. The resulting preprocessing time and space is $O^*(1/\lambda)$.\footnote{The notation $O^*(\cdot)$ hides a $n^{o(1)}$ factor.}

\ignore{
The high level flow of the entire algorithm is shown in \Cref{fig::main_flow}.

\begin{figure}[h!]
\centering
\includegraphics[width=0.98\textwidth]{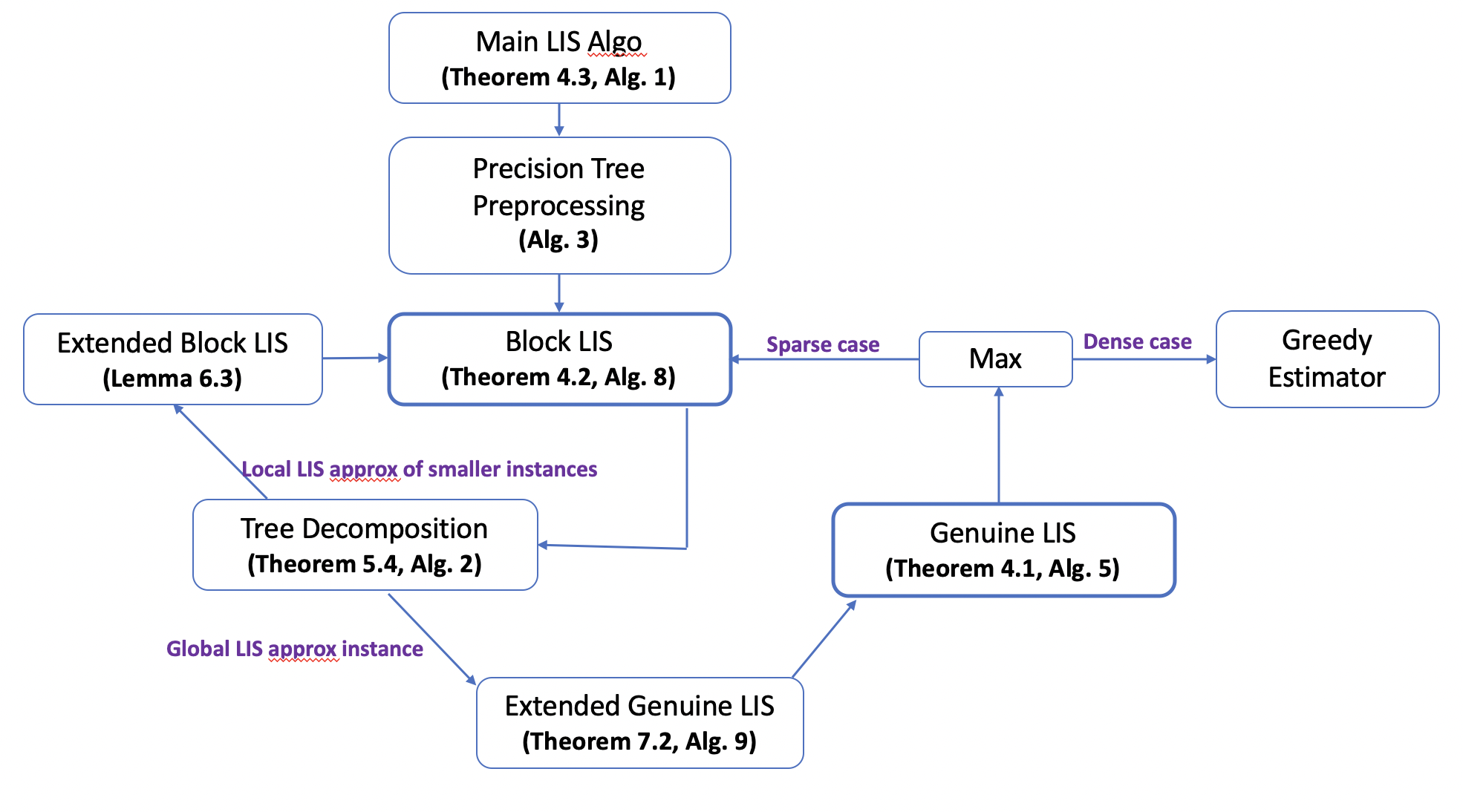}
\vspace{-4mm}
\caption{High level flow of the main algorithm (\Cref{thm::lis_main})}
\label{fig::main_flow}
\end{figure}
}

\section{Main Algorithm for Estimating \LIS} \label{sec:technical:main}

Our main algorithm is composed of two algorithms, for solving $\blocklis$ and $\genlis$. These algorithms recursively call each other, with access to a \prectree data structure. This data structure queries the input sequence at the beginning (non-adaptively and non-uniformly), and our algorithms access the sequence only via this data structure. We describe the details of \prectree in \Cref{sec::precision}, and for now refer to it as a tree $T_\lambda(y)$ for some parameter $\lambda< 1$ and an input string $y$.

We now state the main guarantees of the two algorithms, whose descriptions and proofs will appear in later sections. 
The statements below are built to support mutual recursion and might be challenging to parse at first read. At high-level, the idea is to create mutual inductive hypotheses which can be instantiated with different parameters, in a way that enables concurrent optimization of the approximation, sample complexity and run-time complexity.
We then show how careful recursive instantiations of these two algorithms yield our main algorithm for estimating $\LIS$. 

Below, for an instance $g \in (\N \times \zo)^{n \times k}$, let $g$ restricted to first/second coordinate be $g(*)_1\in\N^{n\times k}$ and $g(*)_2\in\{0,1\}^{n\times k}$ respectively. 



\begin{theorem}[$\estgenlis$ algorithm; 
\Cref{sec::genlis}]
\label{thm::gl_rec_tree}
	Fix integers $n,k$ and $\lambda \in (1/n,1)$. Fix an instance $g \in (\N \times \zo)^{n \times k}$. For some monotone functions $a_s : \R^3_+ \rightarrow [1,\infty)$ and $c_s : \R_+^4 \rightarrow [1,\infty)$, 
	 suppose there exists a randomized algorithm $\cA_{BL}$ that, given $X\subseteq[n],Y \subseteq \N$, parameter $\tau_s$, and a \prectree $T_{1/c_s\left(|X|,\tau_s,\lambda,\tfrac{\lambda\cdot |X|}{|y(X, Y)|}\right)}(y)$ for some $y \in \N^{n \times k}$, can produce an $(\alpha_s, \lambda |X|)$-approximation for $\blocklis(y,X,Y)$ where 
	$\alpha_s= 
	a_s\left(|X|,\tau_s,\tfrac{\lambda\cdot |X|}{|y(X, Y)|}\right)$ w.h.p. in time $t_s$. \ns{$\polylog$ parameter can be improved if needed}

	Fix any parameter $\gamma, \tau \geq 1$, and let $c \geq \tO(1/\lambda)
	+ c_s\left(n,\tau,\lambda,\gamma \lambda / k\right)$. Then, there exists an algorithm $\cA$, that, given free access to $g(*)_1$ 
	and \prectree access to $g(*)_2$, 
	$T_{1/c}(g(*)_2)$,
	produces a $\left(\alpha_g, \lambda n\right)$-approximation for $\genlis(g)$ w.h.p., where
	$\alpha_g=\log(k/\lambda) \cdot \max\{\gamma,a_s(n,\tau,\gamma \lambda / k)\}$. The algorithm $\cA$  
	 runs in time 
	$\tO(n k) + O(t_s)$.

\end{theorem}




\begin{theorem}[$\estreslis$ algorithm; 
\Cref{sec::reslis}]\label{thm::sl_rec_tree}

Fix monotone 
functions $t_s,a_g :  
\R_+^2 \rightarrow [1,\infty)$, $a_s,c_g:\R_+^3 \rightarrow [1,\infty)$ and $c_s:\R_+^4 \rightarrow [1,\infty)$, satisfying, for all $r, \tau \in E_\beta$, $m \in \N$, $\lambda < 1$, $\lambda_1, \lambda_2 \in [\lambda,1]$, and $k' \in [1,1/\lambda]$ with $\lambda_1 \lambda_2 =
	 \Omega(\lambda k' 
	 )$: 

	\begin{itemize}
	\item
$a_s\left(r,\tau,\tfrac{\lambda}{m}\right) \geq \polylog \left(\frac{k}{\lambda}\right) \cdot a_g\left(\tau, \frac{\lambda_2}{k'}\right) \cdot a_s\left(\frac{r}{\tau},\tau,\frac{\lambda_1}{m}\right)$;
\item
$c_s\left(r,\tau,\lambda,\tfrac{\lambda}{m}\right) \geq  
\beta^{O(1)} \cdot \frac{\tau}{\lambda} + c_g(\tau,\lambda_2,k') \cdot c_s\left(\frac{r}{\tau},\tau,\Theta(\lambda_1/k'),\frac{\lambda_1}{m}\right)$; and,
\item
$t_s\left(r,\tau\right) \geq \log^{O(\log_\beta(\tau))}(\beta)  \cdot t_s\left(\tfrac{r}{\tau},\tau\right)$.

	\end{itemize}
	
Suppose there exists 
and algorithm $\cA_{GL}$ with the following guarantee: given $\genlis$ instance $g \in (\N \times \zo)^{n_g \times k_g}$ 
with $\beta$-ary \prectree $T_{1/c_g(n_g,\lambda_g,k_g)}(g(*)_2)$ access, $\cA_{GL}$ outputs  $\left(a_g\left(n_g,\tfrac{\lambda_g}{k_g}\right),\lambda_g n_g\right)$-approximation to $\genlis(g)$ in time 
$\tO(n_g \cdot k_g)$ w.h.p.


Now fix input $y \in \N^{n \times k}$, a block interval $X \subseteq [n]$, value range interval $Y \subseteq \N$, parameters $\lambda \in (0,1)$, $\beta \in \N$, $\tau \in E_\beta$. 
	 Then, 
	 given a $\beta$-ary \prectree $T_{1/c}(y)$, 
	 we can produce a $(\alpha, \lambda |X|)$-approximation for $\blocklis(y,X,Y)$ w.h.p., as long as
$\alpha \geq a_s\left(|X|,\tau,\tfrac{\lambda |X|}{|y(X,Y)|}\right)$ and
$c \geq c_s\left(|X|,\tau,\lambda,\tfrac{\lambda |X|}{|y(X,Y)|}\right)$.

		The algorithm's expected run-time is at most $
	c
	\cdot t_s(|X|,\tau) \cdot \tfrac{ |y(X,Y)|}{ |X|}$.
\end{theorem}

The proofs are deferred to \Cref{sec::genlis} and \Cref{sec::reslis}. 
Combining the two algorithms from above, we obtain the following theorem.

\begin{theorem}\label{thm::reslis_main}
	Fix any $\lambda = o(1)$ and $\epsilon < 1$. 
	There exists a randomized non-adaptive algorithm 
	that solves 
	the $\LIS$ problem with $(\alpha,\lambda n)$-approximation, where $\alpha=(1/\lambda)^{\sqrt{\eps}}\cdot (\log 1/\lambda)^{2^{O(\log^2 (1/\eps)/\sqrt{\eps})}}$ using $\tfrac{1}{\lambda}\cdot n^{O(\sqrt{\eps}\log1/\eps)}$ time (and hence samples from the input).
\end{theorem}

The algorithm for \Cref{thm::reslis_main} 
 (see \Cref{alg::ALG_lis_main}) follows the outline from \Cref{fig::main_flow}. In particular, we first build a $\beta$-ary precision tree $T_{1/c}$ for $\beta = \Theta(\log n)$, and $c = c_s(n,n^\eps,\lambda,\lambda) = \tfrac{1}{\lambda}\cdot n^{O(\sqrt{\eps}\log1/\eps)}$. We then apply the algorithm of \Cref{thm::sl_rec_tree} by interpreting the string as a \reslis instance with $k=1$. The proof 
follows from recursive implementation of
\Cref{thm::gl_rec_tree,thm::sl_rec_tree} 
with carefully chosen parameters. Most importantly $\tau$ and $\gamma$ are carefully chosen, as a function of the other parameters, to balance approximation and complexity. Informally, we pick $\tau \approx |X|^{\eps}$ in each $\reslis$ instance, and $\gamma \approx (k/\lambda)^{\sqrt{\eps}}$ for a recursion of depth $\approx 1/\sqrt{\eps}$, then we stop and use the dense estimator only by setting $\gamma$ to be maximal. 
The formal proof 
involves tedious calculations, and
is deferred to \Cref{apx:proofRecursion}. 
We now complete
\Cref{thm::lis_main}, by instantiating the $\LIS$ problem with the parameters above set suitably. 


\begin{figure}[h!] 
\setstretch{1.2}
  \begin{algorithm}[H]
    \caption{\textsc{EstimateLISMain}}
    \label{alg::ALG_lis_main}
		\DontPrintSemicolon
		\SetNoFillComment
		\KwIn{
		A sequence $y \in \N^n$,
		error parameter $\lambda \in (1/n,1)$, and
		$\epsilon < 1$.}	
		\KwOut{An integer $\widehat{L} \in \left[0, \LIS(y) \right]$.
		}
    $T \gets \textsc{PreprocessPrecisionTree}(y,\lambda \cdot n^{-\eps}, \beta)$, where $\beta \gets \Theta(\log n)$. \\
     \Return $\widehat{L} \gets \estreslis(T,n^\eps,\lambda)$.
\end{algorithm}
\caption{Description of the algorithm \textsc{EstimateLISMain}.}
\end{figure}

\begin{proof}[Proof of 
\Cref{thm::lis_main}]
	Let $y\in \N^n$ be an input of the $\LIS$ problem. 
	We solve 
	the problem
	using the algorithm of 
	\Cref{thm::reslis_main} with inputs $y,\lambda$, and $\eps = 1/\log\log 1/\lambda$, noting that $\alpha=\lambda^{o(1)}$ and runtime complexity is $\tfrac{1}{\lambda}n^{O(1/\sqrt[3]{\log\log 1/\lambda})}=\tfrac{1}{\lambda}n^{o(1)}$. Finally, we invoke \Cref{ft::approx} with $q=o(1)$ to obtain the claimed approximation (as $q = o(1)$, $p = q / (1 - q) = o(1)$).
\end{proof}

\section{
\prectrees Data Structure}\label{sec::precision}

In this section we present the \prectrees data structure used to improve our bounds, and establish the properties of this data structure. In particular, we describe how to decompose the tree so that we can run various sampling-based algorithms on sub-trees of the tree, and recompose these results using yet another sampling-based algorithm.

The basic construction of the $\beta$-ary tree $T_\delta(e)$ was introduced in \Cref{sec:constPT}. In particular, each node $v$ has a precision score $\Pn_v$, defined as follows. For the root, we define 
 $\Pn_{\text{root}(T)} = \delta$.   The precision score of a node $v$ is defined recursively as a function of the parent: 
\begin{equation}
    \Pn_v = \bZ_v \cdot \Pn_{\parent(v)} \text{ where } \bZ_v \sim_{i.i.d.} \uniform(\{1,2,3,\ldots,\beta/4\})
    \ .
    \label{eq:precision-tree}
\end{equation}

For a node $v$, we recurse into its children only if $\Pn_v \le 1$; otherwise we stop the recursion for $v$ with $\Pn_v>1$ (i.e., $v$ has 0 children in the precision tree). If $v$ is a leaf, with the associated integer $i\in [n]$, then we sample the corresponding element $e_i$ of the input string $e$. 

We first establish the sample complexity of the \prectree. In the below, we use $V_\ell$ to denote the set of nodes at level $\ell$.


\begin{lemma}[Sample complexity]\label{lm::sample_bound}
	The expected number of coordinates in $[n]$ that are sampled is 
 $\tfrac{O(\log\beta)^{\log_\beta n}}{\delta}$. 
\end{lemma}

\begin{proof}
    Let $s_v = 1/\cP_v$. Since we only sample leaves with $s_v \geq 1$, then the total sampled leaves is at most $\sum_{v \in V_{\log_\beta(n)}} s_v$.
    
	Now, since $\bZ_v$ are chosen i.i.d. uniformly in $\{1,\ldots,\Omega(\beta)\}$, then we deduce the recursive relation:
	
	\begin{equation}
	\label{eq:leaf-sample}
	\E[s_v] \leq \E\left[\tfrac{1}{\cP_v}\right] = \tfrac{O(\log \beta)}{\beta} \cdot \E\left[\tfrac{1}{\cP_{\parent(v)}}\right]
	\end{equation}
 
	Let $\tau_\ell \triangleq \sum_{v \in V_\ell} s_v$.
	Then $\tau_0 = 1/\delta$, and for level $\ell$, summing and using 
	\Cref{eq:leaf-sample},
	we obtain
	$$
	 \E\left[\tau_\ell\right] = \E\left[\sum_{v\in V_\ell} s_v\right] 
	 \leq  \beta \cdot \E\left[\sum_{v \in V_{\ell-1}} \tfrac{O(\log \beta)}{\beta} \cdot s_v\right] 
	 \leq O(\log \beta)\E\left[\sum_{v \in V_{\ell-1}} s_v\right] 
	 = O(\log \beta)\tau_{\ell-1}. 
	 $$
	 
	 Then $\E[\tau_\ell] \leq O(\log \beta)^\ell/\delta$, and we conclude $\E\left[\sum_{v \in V_{\log_\beta (n)}}s_v\right] = O(\log(\beta))^{\log_\beta n}
	 / \delta$ as needed.
\end{proof}

\subsection{Tree Sampling Data Structure}

We equip the \prectree with a {\em Sampling Oracle Data Structure} (denoted SODS) that should be seen as a data structure wrapper with access to \prectree, allowing {\em efficient sampling} of leaves. 




First, we show one can simulate uniform samples of elements given a \prectree access.

\begin{lemma}[Simulating Random Samples]\label{lm::simulate_samples}
	Fix 
	$\delta \leq 1$ and
	\prectree $T=T_{\delta}(e_1,\ldots,e_n)$, as well as 
	arbitrary $\delta' \geq \delta$, with $1/\delta'\in \N$. Using access to $T$ only, we can generate a set of elements $S \subseteq \{e_1,\ldots,e_n\}$ such that the distribution of $S$ is identical to the distribution where each $e_i$ is included i.i.d. with probability $1/\delta' n$. The runtime to generate $S$ is $O(\log \beta)^{\log_\beta n}/\delta$.
\end{lemma}
We remark that, while two subsequent invocations of the above lemma may give different sets $S,S'$, each with the above distribution, they are dependent between each other.
\begin{proof}
	The main task here is to show independence, i.e., that we can choose a set of i.i.d. samples by reusing the randomness of the \prectree. We prove by induction on the tree height, starting with single node trees, that for any node with $\cP_v\le 1$, for any $\delta'\ge\cP_v$ we can generate a subset of the leaves where each leaf is included with probability at least $1/\delta'n_v$, where $n_v$ to be the number of leaves in the subtree rooted at $v$.

	For leaves, if $\cP_v \leq 1$, the node is included in the \prectree, and we can subsample it with any required probability as needed.
	Now consider a non-leaf node $v$; by induction the statement holds for its children.	
	
	If $\delta' \leq \tfrac{4}{\beta}$, we are basically done since any child $u$ of $v$, has score $\cP_u \leq \tfrac{\cP_v \beta}{4} \leq \tfrac{\delta' \beta}{4}\le1$, while the number of elements $n_u = \tfrac{n_v}{\beta}$, so we can simply subsample leaves in $u$'s tree, each with probability $\tfrac{1}{\delta' n_v} = \tfrac{1}{\delta' \beta n_u} \leq \tfrac{1}{\cP_u n_u}$ as needed (note that $\delta'n_v/n_u=\beta\delta'\ge \cP_u$ as required by the inductive hypothesis).
	
	It remains to show this for $\delta' > \tfrac{4}{\beta}$. 
	For this case we claim we would like to use the tree-randomness to generate the samples.
	
	For distributions $P, Q$, 
we say that $P$ 
{\em stochastically dominates} $Q$ (also denoted $P \succ Q$) 
if $\cF_P(t) \leq \cF_Q(t)$ for all $t$, with strict inequality for some $t$, where $\cF_{P}$ and $\cF_{Q}$ are the cumulative distribution functions (CDFs) of $P$ and $Q$ respectively. 
We first claim the following:
	
	\begin{claim}\label{cl::binomial_dominance}
		Fix $k \in \N$, $0 < p < 1$
		and $F \geq 1$ such that $F\cdot pk < 1/4$. 
		Let $\bX_1,\ldots,\bX_k \sim_{i.i.d.} \Ber(p)$ (i.e., i.i.d. Bernoulli random variables with bias 
		$p$) and let $\bX = \sum_i \bX_i$. Let $\bY_1,\ldots,\bY_k \sim_{i.i.d.} \Ber(2 \, p \, F)$ and let $\bY = \bY^* \cdot \sum_i \bY_i$ where $\bY^* \sim \Ber(1/F)$, independent of the other variables.
		Then, $\bY$ stochastically dominates $\bX$.
	\end{claim}
	
    \begin{proof}
        $\bX$ and $\bY$ are both supported on $\{0,1,\cdots,k\}$. So, it suffices to show that $\Pr[\bY \geq i] > \Pr[\bX \geq i]$ for all $i \in [k]$. 
        Note that $2pF \leq 2pFk < 1/2$, so $2pF$ is a valid probability. For $i \in [k]$, let
        \[
        f(i) := \frac{\Pr[\bY = i]}{\Pr[\bX = i]}
        = \frac{\binom{k}{i} \cdot (2 \, p\,  F)^i \cdot (1 - 2 \, p \, F)^{k - i}}{F \cdot \binom{k}{i} \cdot p^i \cdot (1 - p)^{k - i}}
        = 2^i \cdot F^{i-1} \cdot \left(\frac{1-2 \, p\, F}{1-p}\right)^{k-i}.
        \]

        To show that $\Pr[\bY \geq i] > \Pr[\bX \geq i]$ for all $i \in [k]$, it suffices to show that $f(i) > 1$ for all $i \in [k]$. Using the conditions that $2 p F < 1/2$, $p F  k < 1/4$ and $F \geq 1$, along with the fact that $1 - x \geq 2^{-2x}$ for all $x \in [0, 1/2]$, we obtain:
        \[
        f(i) > 2 (1 - 2 \, p\, F)^k \geq 2 \cdot 2^{-4 p  F  k} > 2 \cdot 2^{-1} = 1,
        \]
        for all $i \in [k]$. This concludes the proof.
    \end{proof}
	
	 We show that one can simulate sub-sampling of each leaf of $u$ i.i.d. with probability $1/\delta' n_v$ by the following process. If $\cP_u > 1$, then output no elements. Otherwise, we can (by the inductive hypothesis) sub-sample each leaf of $u$ 
	 i.i.d. with probability $1/n_u$.

	 Now, let
	 \[
	 F \triangleq
	 \tfrac{1}{\left\lfloor1/\delta' \right\rfloor}
	 \tfrac{\beta}{4} \geq
	\tfrac{1}{ \left\lfloor1/\cP_v \right\rfloor}
	  \tfrac{\beta}{4} \stackrel{Eq.~\ref{eq:precision-tree}}{\geq} \tfrac{1}{\Pr_T[\cP_u \leq 1]}.
	  \]
	 This process generates $\mu$ i.i.d. leaves of $u$ with distribution 
	$$
	\mu \sim \Ber(\Pr_T[\cP_u \leq 1]) \cdot \sum_{l \text{ leaf of } u} \Ber(\tfrac{1}{n_u}) \succ \Ber(\tfrac{1}{F}) \cdot \sum_{l \text{ leaf of } u} \Ber(\tfrac{1}{n_u}).
	$$
	
	Noting that $n_u = \tfrac{n_v}{\beta} \geq \tfrac{\delta' n_v}{2F}$ and $\tfrac{n_u F}{\delta' n_v} \leq \tfrac{1}{4}$, we invoke \Cref{cl::binomial_dominance} to obtain that the  distribution $\mu$ {\em stochastically dominates} the required sample-size distribution $\sum_{l \text{ leaf of } u} \Ber(\tfrac{1}{\delta' n_v})$, and hence we can simulate the distribution we need over leaves of each child, and hence also for $v$.

For run-time, we note that this 
process takes time (at most) proportional to the size of the tree. For the latter, by \Cref{lm::sample_bound}, the time is $O(\log \beta)^{\log_\beta n}/\delta$.

%
%
\end{proof}

Next, we show that if each element of the tree is a block of integers, then one can construct a data structure that sub-samples with {\em interval range restriction}, in time proportional to the sample size.

\begin{corollary}[Conditional sub-sampling data structure]\label{cr::sampling_ds}
	Fix  
	\prectree $T=T_{\delta}(y)$ with $y \in \N^{n \times k}$.  
	There exists a data structure, that given any interval 
	$Y \in \cI$ and sub-sampling probability $\eta \leq 1/\delta n$ with $1/\eta\in E_2(n)$, sub-sample block-coordinates $X' \subseteq [n]$ i.i.d. with probability $\eta$ and outputs a set of all coordinates $W \subseteq X' \times [k]$ such that $y_W \in Y$, i.e. $W \triangleq \cup_{i \in X'} \{(i,j) \mid y_{i,j} \in Y \}$,
in time $\tO(|W|)$. 
The preprocessing time is,
in expectation,
$\tO\left(\tfrac{n + \norm{m}_1}{\delta n} 
\right)$,
 where $m_i \leq k$ is the number of non-null elements in $y_i$.
\end{corollary}
\begin{proof}
	To preprocess, we prepare for each possible $\eta\le 1/\delta n$ with $1/\eta\in E_2(n)$ by sub-sampling each block i.i.d. with probability $\eta$ using \Cref{lm::simulate_samples}. Then we compute a set of coordinates $Z_\eta\subset [n]$ by combining all coordinates from all sampled blocks. Let $U_\eta \triangleq \cup_{w \in Z_\eta \times [k]} \{y_w\}$.
	Compute $U_\eta$ and store it as a sorted array with a pointer to $y^{-1}(j)\cap (Z_\eta \times [k])$ for each $j \in U_\eta$.
	Now, for each $(\eta,Y)$ query, locate the range of elements in $U_\eta$ containing precisely the elements in $Y$ using binary search of $\min(Y), \max(Y)$. This will be our output $W$.
	
	The correctness follows immediately from the construction. Preprocessing runtime is as claimed as, in expectation, we need to sort $\|m\|_1/\delta n$ elements, and there are $O(\log n)$ different values of $\eta$ to consider. Runtime is only $O(\log n)$ plus the size of the output (as it is stored as contiguous block).
		\end{proof}

\subsection{Decomposing \prectrees}



We now show how to use the \prectree for recursive sampling-based algorithms via decomposing the tree.
In particular, consider the setting where a top-level/global algorithm decomposes the input string into intervals, and recurses on a sample of these intervals to solve a ``local'' problem on each. If one can compute a local function over each such (sampled) interval given a local \prectree access to such an interval, then we can solve the overall problem by recasting the overall precision tree as a small global precision tree where each leaf is a precision tree by itself. One can observe we can do this for any size intervals, as long as the precision score of local tree $\eta$ and precision score of global tree $\delta$ are such that the original tree has precision score $\eta\delta$ (i.e., it is $T_{<\eta \delta}$).

The algorithm follows the following steps (the full algorithm is described in \Cref{fig:trimtree}):
\begin{enumerate}
	\item Identify the ``correct'' level $\ell$ in $T$ where each node corresponds to the right number of elements.
	\item Find all nodes $v$ at level $\ell$ with enough precision to compute the local function $f_v$ and compute the local function on these nodes by ``detaching'' the sub-trees rooted at each $v$.
	\item Consider the global tree $T$, trimmed to levels $\leq \ell$ (i.e., the {\em top} of the tree). Augment the precision parameters by dividing each score with the precision we need to use for the local computation. This creates another ``simulated'' precision tree where we have access to the value at all leaves with $\cP_v \leq 1$. 
	\item Compute the global function $g$ with tree access to the locally computable values.
\end{enumerate}

Overall, we constructively show the following. Below, we assume that $n$ is an integral power of $\beta$, but the results extend to general tree sizes (for example by padding the input). 

\begin{lemma}[Tree Decomposition]\label{lm::tree_decomposition}
	Fix $\beta,\tau,n \in \N$ such that $n,\tau$ are each a power of $\beta$. Fix a domain of elements 
	$e \triangleq \{e_1,\ldots,e_n\} \in \cE^n$ and let $X_1,\ldots,X_\tau$ be a partition of $[n]$ into disjoint intervals of equal width. 
	Fix $\delta, \eta \in (0,1)$ and 
	(possibly randomized) functions $g: T_{\delta}(\cF^\tau) \rightarrow \cG$ and $\{f_i: T_\eta(e(X_i)) \rightarrow \cF\}_{i \in [\tau]}$ for some spaces $\cG$ and $\cF$.
	Then, algorithm TrimTree (see 
	\Cref{fig:trimtree}), given a \prectree $T_{\delta \eta}(e)$, outputs  $g(T_\delta(f_1(T_\eta(e(X_1))),\ldots,f_\tau(T_\eta(e(X_\tau)))))$. 
	The algorithm's expected run-time is $\tO(\tau + t_g + 
	\tfrac{1}{\delta \tau} \cdot \norm{t_f}_1\cdot \log^{O(\log_\beta(\tau)}(\beta) )$, where $t_g$ is the time to compute $g$, and $(t_f)_i$ is the expected time to compute $f_i$.
	\end{lemma}
	
	\begin{proof}[Proof of 
	\Cref{lm::tree_decomposition}]
	Consider the full \prectree $T_{\delta \eta}(e)$. Take the top $\log_{\beta}\tau$ levels of the tree to get to the level where each node $v$ has $n/\tau$ leaves in its subtree.  This top portion (trimmed tree) can be seen as $T_\delta$ as follows: the precision $\Pn'_v$ in $T_\delta$ is simply defined to be $\Pn_v/\eta$. Hence for leaf $v$ in $T_\delta$, we have that $\Pn'_v\le 1$ iff $\Pn_v\le \eta$. Hence for each such $v$, if $\Pn'_v\le 1$, we have another $T_{\eta}$ rooted at $v$ (in the full \prectree), allowing us to compute the corresponding $f_i(T_\eta(e(X_i)))$. Hence we will be able to also compute $g(T_\delta(f_1(T_\eta(e(X_1))),\ldots,f_\tau(T_\eta(e(X_\tau)))))$ as required.
	

	For runtime, we note that by
	\Cref{lm::sample_bound}, we have $\tfrac{1}{\delta} \cdot \log^{O(\log_\beta(\tau)}(\beta)$ leaves $v$ with $\cP_v' \leq 1$. We only preprocess SODS and compute the function $f_i$ on those leaves, and all $f_i$ have the same probability to be computed. Also, we note from the independence of $\bZ_i$'s, once we fix precision parameter $\eta$, the time to compute each $f_i(T_\eta(e(X_i)))$ is independent of 
	the random choice of $i$ where $f_i$ is computed, hence we have the expected runtime is the product of expectations.
	Finally, computing the function $g$ takes $t_g$ time. 
	
	\end{proof}

%


\begin{figure}[h]
\setstretch{1.2}

  \begin{algorithm}[H]
    \caption{\textsc{TrimTree}}\label{fig:trimtree}
		\DontPrintSemicolon
		\SetNoFillComment
		\KwIn{Disjoint consecutive intervals of $[n]$: $\{X_i\}_{i\in [\tau]}$ of length $r$ 
		each, global function $g: T_{\delta}(\cF^\tau) \rightarrow \cG$, 
		functions $\{f_i : T_\eta(e(X_i)) \rightarrow \cF\}_{i \in [\tau]}$, 
		\prectree $T$, local precision parameter $\eta$.}
		\KwOut{$g([f_i(T_\eta(e(X_i)))]_{i\in [\tau]})$}
		Let 
		$\ell = \log_\beta \tau$ be the level in $T$ such that for $v \in V_\ell$, $|X_v| = r$. \\

		$T_g \gets$ the induced sub-tree of $V_{\leq \ell}$. \\
		\For{${v \in V_\ell}$}
		{
			\If{$\cP_{v} \leq \eta$}{Preprocess SODS for the precision tree rooted at $v$ using \Cref{cr::sampling_ds}. \\
			Compute $f_i(T_\eta(e(X_i)))$ using the sub-tree rooted at $v$ and store the result in $v$.}
		}
		$\cP_v \gets \cP_v / \eta$ for all $v \in T_g$. \\
		Preprocess SODS for $T_g$ using \Cref{cr::sampling_ds}. \\
		Compute $q \gets g([f_i(T_\eta(e(X_i))]_{i\in [\tau]})$ using $T_g$. \\ 
		\Return $q$.
\end{algorithm}
\caption{Algorithm \textsc{TrimTree} for decomposing \prectrees (\Cref{lm::tree_decomposition}). 
Define level $\ell_v$ of a node $v$ as its distance to the root. We define $V_\ell$ as the set of nodes of level $\ell$, i.e., $V_\ell \triangleq \{v \in T \mid \ell_v = \ell\}$. We use the shorthand $V_{\leq \ell} \triangleq \cup_{m \leq \ell} V_m$ and similarly for $V_{<\ell}, V_{>\ell},V_{\geq \ell}$.}
\end{figure}

\begin{figure}[h]
\setstretch{1.2}
  \begin{algorithm}[H]
    \caption{\textsc{PreprocessPrecisionTree}}
		\DontPrintSemicolon
		\SetNoFillComment
		\KwIn{
		An input $e \in \cE^n$,
		precision parameter $\delta$, and
		branching parameter $\beta \in \N$.}	
		\KwOut{\prectree data-structure $T_{\delta}(e)$.
		}
	$T \gets \textsc{PreprocessNode}(e,[|e|],\delta,\beta)$. \\
    Preprocess SODS for $T$ using \Cref{cr::sampling_ds}. \\
     \Return $T$. 
\end{algorithm}

  \begin{algorithm}[H]
    \caption{\textsc{PreprocessNode}}
		\DontPrintSemicolon
		\SetNoFillComment
		\KwIn{
		An input $e \in \cE^n$,
		an interval $I \subseteq [|e|]$,
		precision parameter $\delta$, and
		branching parameter $\beta \in \N$.}	
		\KwOut{\prectree data-structure $T_{\delta}(e)$.
		}
	$v \gets$ a new node associated with $e_I$. \\
	\If {$\delta \leq 1$}
    {
	\If {$|I| = 1$}
	    {Store $e_I$ in $v$.}
	\Else {
	$I_1,\ldots,I_\beta \gets$ a partition of $I$ into $\beta$ disjoint, equal length intervals. \\
	\For {$i \in [\beta]$}
	{
	    $\bZ_i \sim \uniform(\{1,2,3,\ldots,\beta/4\})$. \\
	    $\cP_i \gets \delta \cdot \bZ_i$. \\
	    $u_i \gets \textsc{PreprocessNode}(e,I_i,\cP_i,\beta)$. \\
	    Add $u_i$ as a child of $v$. 
	    }
	 }}
     \Return $v$. 
\end{algorithm}
\caption{Description of the algorithm \textsc{PreprocessPrecisionTree}.}
\label{fig::ALG_preprocess_precision_tree}
\end{figure}

\newcommand{\dsamples}[0]{\tfrac{1}{\lambda}}

\section{Algorithm for \genlis: Proof of \Cref{thm::gl_rec_tree}}\label{sec::genlis}

In this section we describe and analyze the algorithm for the \genlis problem, which is called \estgenlis.

First, we introduce an extension to the \reslis algorithm, which we use for our algorithm. 
In particular, 
we allow it to ignore 
a few ``heavy blocks'' which will not affect the approximation or complexity guarantees of \Cref{thm::gl_rec_tree}. 
We use the following definition:

\begin{definition}
	Fix block sequence 
	$y \in \N^{X \times [k]}$, 
	parameter $\delta < 1$ and value range interval $Y \in \cI$. Let $X^\delta$ be the $\delta|X|$ blocks  
	in $X$ for which 
	the quantity $|y(X^\delta) \cap Y|$ is maximized. The \emph{$\delta$-heavy trimmed input} $y_{-\delta,Y} \in \N^{X \times [k]}$ is $y$ after setting $y_{-\delta,Y}(X^\delta) = \emptyset$ (i.e., removing all integers from $y(X^\delta)$).
	\end{definition}

For $X' \subseteq X$ and a range of $y$-values $Y'$, 
we also define 
$y_{-\delta}(X',Y') \triangleq y(X')_{-\delta,Y'} \cap Y'$ (i.e., the multiset of values in $y(X')_{-\delta,Y'}$ restricted to $Y'$).

The stronger version of \Cref{thm::gl_rec_tree} is as follows:
\begin{lemma}[$\estgenlis$ algorithm, extended]
\label{lm::gl_rec_tree_ext}
	Fix integers $n,k$ and $\lambda \in (1/n,1)$. Fix an instance $g \in (\N \times \zo)^{n \times k}$. For some monotone functions $a_s : \R^3_+ \rightarrow [1,\infty)$ and $c_s : \R_+^4 \rightarrow [1,\infty)$, 
	suppose there exists a randomized algorithm $\cA_{BL}$ that, given $X\subseteq[n],Y \subseteq \N$, parameter $\tau_s$, and a \prectree  $T_{1/c_s\left(|X|,\tau_s,\lambda,\tfrac{\lambda\cdot |X|}{|y_{-\lambda/2}(X,Y)|}\right)}(y)$ 
	for some $y \in \N^{n \times k}$, can produce an $(\alpha_s, \lambda |X|)$-approximation for $\blocklis(y,X,Y)$ where 
	$\alpha_s= 
	a_s\left(|X|,\tau_s,\tfrac{\lambda\cdot |X|}{|
	y_{-\lambda/2}(X,Y)|}\right)$ in time $t_s$, w.h.p.
	
	Fix any parameter $\gamma, \tau \geq 1$. Then, there exists an algorithm $\cA$ that, given complete direct access to $g(*)_1 \in \N^{n \times k}$ and \prectree access to $g(*)_2 \in \zo^{n \times k}$, termed
	$T_{1/c}(g(*)_2)$,
	produces a $\left(\alpha_g, \lambda n\right)$-approximation for $\genlis(g)$ w.h.p., where $c \geq \tO(1/\lambda)
	+ c_s\left(n,\tau,\lambda,\gamma \lambda / k\right)$ and
    $\alpha_g=\log(k/\lambda) \cdot \max\{\gamma,a_s(n,\tau,\gamma \lambda / k)\}$. The algorithm $\cA$  
	 runs in time 
	$\tO(n k) + O(t_s)$.

\end{lemma}
\ns{$\polylog$ parameter can be improved if needed} 
	
To prove \Cref{lm::gl_rec_tree_ext}, we use the following reduction: 

\begin{lemma}\label{lm::reslis_int_reduction}
    Suppose there exists an algorithm $\cA$ that $\left(a\left(\tfrac{\lambda |X|}{|y(X,Y)|}\right),\lambda n\right)$-approximates $\reslis$ w.h.p. given tree access $T_{1/c}(g)$, in time $t=t\left(\tfrac{\lambda |X|}{|y(X,Y)|}\right)$, where $c=c\left(\tfrac{\lambda |X|}{|y(X,Y)|}\right)$  and $a,c,t$ are some functions. Then there exists an algorithm $\cA'$ that $\left(O\left(a\left(\tfrac{\lambda |X|}{|y_{-\lambda/2}(X,Y)|}\right)\right),\lambda n\right)$-approximates $\reslis$ with tree access $T_{1/c\left(\tfrac{\lambda |X|}{|y_{-\lambda/2}(X,Y)|}\right)}$ in time $O(t) + \tO(1/\lambda)$ w.h.p.
\end{lemma}

The proof of \Cref{lm::reslis_int_reduction} is deferred to \Cref{sec:genlis_ext}.

\begin{proof}[Proof of \Cref{thm::gl_rec_tree} using \Cref{lm::gl_rec_tree_ext,lm::reslis_int_reduction}]
The proof is immediate since the 
extended algorithm assumed in \Cref{lm::gl_rec_tree_ext} reduces to the standard versions using \Cref{lm::reslis_int_reduction}.
\end{proof}

The rest of this section is devoted to prove 
\Cref{lm::gl_rec_tree_ext}.




\subsection{Algorithm Description}



The algorithm $\estgenlis$ is presented in \Cref{fig:ALG_genlis}.
It consists of the following high-level steps. First, we perform some preprocessing by extracting {\em pseudo-solutions}, longest increasing subsequences, ignoring the genuineness flags. Then we partition these sequences according to their lengths. Next we estimate the proportion of genuine elements amongst the union of above subsequences; this gives our first (dense) estimator. We also employ a second (sparse) estimator, when the proportion of genuine elements is too low, to reduce the instance to \blocklis\ with improved parameters. We now describe these steps in detail.

\subsubsection{Extracting pseudo-solutions}

For this step, we treat all elements 
the same, whether they are genuine or not, and greedily compute and extract disjoint increasing subsequences, each of length  $\gtrapprox \lambda n$, 
denoted $\{P_i\}_i$ (which are called 
{\em pseudo-solutions}), until there is no longer an increasing subsequence of length $\lambda n$ remaining. For reasons that will become clear later, we also need to guarantee that each time we extract a sequence, such a sequence is approximately 
proportional to the longest one at that time, i.e., {\em after removing all previous sequences}, and therefore we use a greedy algorithm.

%
%
%



We note that we cannot afford to use the 
standard dynamic-programming solution for finding an optimal $\LIS$ each time, as the resulting runtime is too large for us. To improve the runtime, we use instead a fast data-structure for approximate greedy extraction, based on the fully dynamic data-structure 
in \cite{GJ21_dynamic_LIS}.

\begin{theorem}\cite{GJ21_dynamic_LIS}\label{thm:dynamic_LIS}
	Given an input $x \in \N^n$, for any $\eps > 0$, 
	there exists a {\em dynamic} data-structure $\cD$ with the following 
	operations: 1) inserting/deleting an element, in time $(\log n / \eps)^{O(1)}$, and 2) finding an approximate longest increasing sub-sequence of length within $(1+\eps)$-factor of the optimum, $\OPT$, in time $O(|\OPT|)$.
\end{theorem}


This theorem with $\eps=1$ leads to the following corollary for \reslis.

\begin{corollary}[Extracting Pseudo-solution, approximate]\label{thm::dynamic_main}
    There exists a data-structure that given a sequence of integer blocks $y \in \N^{n \times k}$, extracts (i.e., finds and removes) an increasing sequence of length at least $|\OPT|/2$, 
    where $\OPT$ is the current $\LIS$. The pre-processing time is $\tO(n \cdot k)$.
\end{corollary}

\begin{proof}
    To preprocess the data-structure, we insert all elements to $\cD$ one by one, where elements of each block are inserted in non-increasing order. This takes time $\tO(nk)$. Finally, we answer the extraction queries by $\LIS$-querying $\cD$, 
    and remove all elements in the $2$-approximate $\LIS$ solution 
    one by one from $\cD$.
\end{proof}




\subsubsection{Partitioning the pseudo-solutions}
Next, we partition $\{P_i\}_i$ 
into $\log (1/\lambda)$ disjoint buckets denoted $\{\cP_\ell\}_\ell$, where for $\ell\in E_2(1/\lambda)$, each $P \in \cP_\ell$ is an increasing sub-sequence of length $\approx \ell \lambda n$ (again, ignoring the genuineness of the elements). 
Our algorithm generates an estimator for $\genlis$ 
for each scale $\ell$, 
based purely on the elements of the subsequences $\{P \mid P \in \cP_\ell\}$, 
and outputs the maximal estimator.
Clearly, one of the buckets 
must contain a significant fraction of the $\genlis$; we focus on that bucket $\cP_\ell$ henceforth.

\subsubsection{Estimator Computation}
We next describe our estimator for each bucket as above, which by itself is generated by taking the maximum of 2 estimators, denoted {\em Dense Estimator} and {\em Sparse Estimator}. 
The {\em Dense Estimator} is a fairly straightforward one. 
Here, we consider the union of all pseudo-solutions $U_\ell \triangleq \cup_{P \in \cP_\ell} P$, and estimate the rate of genuine elements in $U_\ell$ by
checking the genuineness of a random sample of elements in $U_\ell$. The dense estimator is then proportional to the estimated rate of genuine elements in $U_\ell$.
The main idea is that, if 
$U_\ell$ has 
many genuine elements, then the aforementioned sampling procedure gives a high-fidelity estimate.

The sparse estimator handles the opposite situation: when the genuine elements are sparse within the union $U_\ell$ of subsequences.
Then, if the \genlis itself is quite large, it must be the case that a high fraction of the genuine elements participate in the \LIS. Here, we invoke the algorithm for \reslis, restricted to genuine elements only. 
Using the fact that \reslis approximation is a function of the proportion of $\LIS$ elements among all the 
``relevant'' elements 
(in this case, genuine elements), we make progress in the sparse case 
by reducing the number of relevant elements, 
resulting in improved approximation.
%
The overall estimator for length $\ell$ is the maximum of the $2$ estimators.

We highlight the fact that the sparse estimator calls the assumed algorithm $\abl$ on the input string $g$, restricted only to the genuine elements. In this call, we do not compute the genuineness flags explicitly, but rather compute them as needed, i.e., only for the elements that are ``read'' by the $\abl$ algorithm, and if an element turns out to be non-genuine it becomes a ``null'' entry. Since we are using \prectree to access samples, the algorithm will perform computation in a different (but equivalent) order: first computing the genuineness of the right samples, and then run the $\abl$ algorithm on those samples.

\begin{figure}[h!]
    \setstretch{1.2}
  \begin{algorithm}[H]
    \caption{\estgenlis}\label{fig:ALG_genlis}
		\DontPrintSemicolon
		\SetNoFillComment
		\KwIn{$g \in (\N \times \zo)^{n \times k},\lambda\in [1/n,1]$ 
		\prectree $T= T_{1/c}(g(*)_2)$ 
		for the parameter $c$ defined in 
		\Cref{lm::gl_rec_tree_ext}.}
		\KwOut{A number $\widehat{G} \in \left[0,n\right]$.
		}
		
		Greedily and iteratively extract 2-approximate maximal {\em increasing sequences} from $g(*)_1$, as long as they contain at least $\lambda n / 4$ elements using \Cref{thm::dynamic_main}, generating pseudo-solutions $P_1,\ldots,P_t$ for some $t = O(k/\lambda)$, where each $P_i \subseteq [n] \times [k]$ is a set of coordinates.\label{alg_gen_LIS_pseudo}\\
		
    Sample blocks $S \subseteq [n]$ i.i.d. with probability $10 \zeta\cdot \dsamples \cdot \tfrac{1}{n}$ each, for $\zeta = O(\log n)$, using \Cref{cr::sampling_ds} and $T$.\\
    Compute 
    $W \gets \{(i,j) \in S \times [k] \mid g(i,j)_2 =1\}$ the set of genuine (sampled) coordinates.\\
    
        \For{$\ell \in E_2(4/\lambda)/4$}{
     Let $\cP_\ell \triangleq \{P_i \, : \, |P_i| \in (\lambda \, n \, \ell/2,\lambda \, n \, \ell]\}$, 
     and let $U_\ell \triangleq \cup_{P \in \cP_\ell} P$. \\
     Compute $B_\ell \subseteq S \gets$ the $\zeta$ 
    blocks in $S$ for which 
    $|(B_\ell \times [k]) \cap W \cap U_\ell|$ is maximized.\\
    $\kappa_\ell \gets |W \cap U_\ell \cap ((S \setminus B_\ell) \times [k])|$. \\
    
    }
    \Return 
    \[
    \widehat{G} \triangleq \max_\ell \max\left\{\tfrac{\lambda^2 \, \ell \, n \, \kappa_{\ell}}{100 \cdot k \cdot \zeta} - \tfrac{\lambda n}{4 },\abl\left(g(*)_1 \mid_{U_\ell\cap (g(*)_2)^{-1}(1)},
    \N,\lambda,T \right)\right\}.
    \]
\end{algorithm}
\caption{Description of the algorithm $\estgenlis$. $g(*)_1\mid_{U_\ell\cap (g(*)_2)^{-1}(1)}$ denotes the sequence $g(*)_1$ restricted to positions in $U_\ell$ for which $g(\cdot)_2=1$. Recall that 
$E_2(k) \triangleq \left\{2^i \mid i \in \{0,\ldots,\lfloor\log
k\rfloor \}\right\}$. $\cA_{BL}$ is the assumed algorithm for \blocklis.
}
\end{figure}

\subsection{Analysis}
We now analyze the algorithm, proving 
\Cref{lm::gl_rec_tree_ext}.

\begin{proof}[Proof of 
\Cref{lm::gl_rec_tree_ext}]

\newcommand{\hOPT}{{\widehat{\OPT}}}
Let \OPT 
be the coordinates of an optimal solution of length $\genlis(g)$. 
Let $U \triangleq \cup_\ell U_\ell$, and $\overline{U} \triangleq [n] \times [k] \setminus U$. 
Define $\hOPT \triangleq \OPT \cap U$ (i.e., restricted to elements in $U$). 
Since $\overline{U}$ does not contain an increasing sequence of length $\tfrac{1}{2}\lambda n$, 
we have that $|\hOPT| \geq |\OPT| - |\OPT \cap \overline{U}| \geq |\OPT| - \tfrac{1}{2}\lambda n$. For $\ell \in E_2(4/\lambda)$, let us also define $\hOPT_\ell \triangleq \OPT \cap U_\ell$.

To bound our estimators, we first bound the quantity $\kappa_\ell$. For this task, we introduce more notation.  
Let $M \subset [n]$ be the 
$\lambda n /2$ blocks containing the highest number of genuine elements. 
We also 
define the following:

\begin{align*}
 \widehat{\kappa_\ell} &\gets |W \cap U_\ell|. \\
  q^{(\ell)} &\triangleq \tfrac{1}{nk} \cdot |U_\ell \cap (g(*)_2)^{-1}(1)|. \\
  q_{-\lambda/2}^{(\ell)} &\triangleq \tfrac{1}{nk} \cdot |U_\ell \cap (g(*)_2)^{-1}(1) \cap (([n] \setminus M) \times [k])|. \\
m_{-\lambda/2}^{(\ell)} &\triangleq \tfrac{1}{k} \cdot \min_{i \in M}\{|U_\ell \cap (g(*)_2)^{-1}(1) \cap (i,*)|\}. \\
G_{\ell,1} &\triangleq \tfrac{\lambda^2 \, \ell \, n \, \kappa_{\ell}}{100 \cdot k \cdot \zeta} - \tfrac{\lambda n}{4 }. \\
G_{\ell,2} &\triangleq \text{ the output of } \abl\left(g(*)_1 \mid_{U_\ell \, \cap \,  (g(*)_2)^{-1}(1)},
    \N,\lambda,T \right). 
\end{align*}
%
%

We have 
\[
    \E\left[\widehat{\kappa_{\ell}}\right]
    = q^{(\ell)} \cdot k \cdot \dsamples\cdot 10\zeta = q^{(\ell)} \cdot \tfrac{10 k \cdot \zeta}{\lambda}.
\]

We need to have a good estimator w.h.p. and therefore use the following statistical fact:
	\begin{fact}\label{ft::stat_max}
		Fix $\zeta \in \N 
		$. 
		Let $\bX_1,\ldots,\bX_m$ be non-negative 
		independent random variables and let $\bX = \sum_i \bX_i$. Let $\bY$ be the sum of the $\zeta$ largest (empirical) $\bX_i$. Then, $\bX - \bY \leq 2\E[\bX]$ with probability $1 - \exp(-\Omega(\zeta))$.
	\end{fact}
	\begin{proof}
		Define $p_i \triangleq \Pr[\bX_i > 2\E[\bX]/\zeta]$, and notice that $\sum_i p_i < \zeta/2$ using Markov's inequality. Define also $\bZ_i \triangleq \max\{0,\bX_i - 2\E[\bX]/\zeta\}$, and $\bZ = \sum_i \bZ_i$. 
		Then $\bX - \bZ$ is a sum of non-negative independent random variables bounded by $2\E[\bX]/\zeta$, and hence, by the Chernoff bound (multiplicative form), with probability at most $\exp(-\Omega(\zeta))$, 
		we have $\bX - \bZ > 2\E[\bX]\geq \left(1 + \tfrac{\E[\bX]}{\E[\bX - \bZ]}\right)\cdot \E[\bX - \bZ]$. 
		Also, notice that $\Pr[\bZ_i > 0] = p_i$, and hence, with probability at most $\exp(-\Omega(\zeta))$, $\sum_i \1[\bZ_i > 0] > \zeta \geq \sum_i p_i + \zeta/2$, implying that $\bZ > \bY$ with this  probability as well. 
		By the union bound, we conclude that $\bX - \bY \leq 2\E[\bX]$ with probability $1 - \exp(-\Omega(\zeta))$ as needed.
	\end{proof}
	

Now, notice that $\widehat{\kappa_{\ell}}$ is a sum of independent random variables and hence from \Cref{ft::stat_max}, w.h.p. we have 
\[
\kappa_{\ell} \leq 2\E[\widehat{\kappa_\ell}] = 20 q^{(\ell)} \cdot \tfrac{k \cdot \zeta}{\lambda}.
\]


For the lower bound, we derive a lower bound for $\kappa_\ell$ using $q_{-\lambda/2}^{(\ell)}$.
Notice that each $i \in M$ gets sampled i.i.d. with probability $\tfrac{10 \zeta}{\lambda n}$. Hence $|S \cap M| > 2\zeta$ w.h.p,  which implies that $|(S \setminus B_\ell) \cap M| > \zeta$. This means that the contribution of blocks in $M$ to $\kappa_\ell$ is at least $\zeta \cdot k \cdot m_{-\lambda/2}^{(\ell)}$.
Next, 
consider the contribution $\psi_\ell$ of blocks in $[n] \setminus M$ to $\kappa_\ell$. Then $\psi_\ell$ is a sum of independent random variables in $k \cdot [0,m_{-\lambda/2}^{(\ell)}]$ and hence, using the Chernoff bound, we have w.h.p.
\[
\psi_\ell \geq 0.9\E\left[\psi_{\ell}\right] - \zeta \cdot k \cdot m_{-\lambda/2}^{(\ell)} 
= \tfrac{9 \zeta}{\lambda}\cdot k \cdot q_{-\lambda/2}^{(\ell)} - \zeta \cdot k \cdot m_{-\lambda/2}^{(\ell)}.
\]

Therefore $\kappa_{\ell} \geq \psi_\ell + \zeta \cdot k \cdot m_{-\lambda/2}^{(\ell)} 
\geq \tfrac{9 \zeta}{\lambda} \cdot k \cdot q_{-\lambda/2}^{(\ell)}$ w.h.p. as well. 
Hence we can finally bound $G_{\ell,1}$, w.h.p., as follows:
 \[
    G_{\ell,1}
    = \tfrac{\lambda^2 \, \ell \, n \, \kappa_{\ell}}{100 \cdot k \cdot \zeta} - \tfrac{\lambda n}{4 }
    \in \left[\tfrac{9}{100} \lambda n \ell \cdot q_{-\lambda/2}^{(\ell)} - \tfrac{\lambda n}{4}, \, \tfrac{1}{4} \lambda n \ell \cdot q^{(\ell)}\right].
\]

For the second 
estimator, we have that $
    G_{\ell,2}$ is a $\Big(a_s\Big(n,\tau,\tfrac{\lambda n}{|U_\ell \setminus (M \times [k]) \, \cap \, g(*)^{-1}(1)|}\Big),\lambda n\Big)$-approximation of \\ $\genlis(g(U_\ell))$.

We now bound the maximum over these quantities.

\subsubsection{Upper Bound}
W.h.p., for all $\ell$, we have $G_{\ell,1} \leq \tfrac{1}{4} \cdot \lambda \, n \, \ell \, q^{(\ell)}$. 
Now, by construction, 
all 
pseudo-solutions $P \in \cP_\ell$ are of length $|P| \geq \lambda \, n \, \ell / 4$, 
and hence $|\cP_\ell| \leq \tfrac{4 k}{\lambda \ell}$. Therefore, there must exist some increasing sub-sequence $P^* \in \cP_\ell$
such that $|P^* \cap (g(*)_2)^{-1}(1)| \geq \tfrac{1}{4 } \cdot \lambda \, n \, \ell \, q^{(\ell)}$, 
i.e., that contains at least that many genuine elements. This implies $|\hOPT_\ell| \geq G_{\ell,1}$.

Similarly, for $G_{\ell,2}$, the $\reslis$ approximation algorithm outputs a lower bound on the $\LIS$, and hence $G_{\ell,2} \leq |\hOPT_\ell|$.

We conclude that $\max_{\ell}\max\{G_{\ell,1},G_{\ell,2}\} \leq \max_\ell |\hOPT_\ell| \leq |\OPT|$.

\subsubsection{Lower Bound}
Since $\hOPT = \cup_\ell \hOPT_\ell$,
 there exists $\ell^* \in E_2(4/\lambda)$ such that $|\hOPT_{\ell^*}| \geq \tfrac{|\hOPT|}{\log (4/\lambda)}$. 
Fix $U^* \triangleq U_{\ell^*} \setminus (M \times [k])$, 
$q^* = q^{(\ell^*)}_{-\lambda/2}$, $\hOPT^* \triangleq \hOPT_{\ell^*}$, 
$G^*_1 = G_{1,\ell^*}$ and $G^*_2 = G_{2,\ell^*}$.

Consider the following two cases:

\paragraph{Dense Case:}
First, suppose $q^* \geq \tfrac{1}{\gamma}$. Observe that from 
\Cref{thm::dynamic_main}, we may only extract a solution of length at most 
$\lambda n \ell^*$ once the remaining $\LIS$ is at most 
$2\lambda n \ell^*$ (i.e., after extracting longer subsequences). Hence $|\OPT_{\ell^*}| \leq 2\lambda n \ell^*$.
Therefore,
\[
G^*_1 \geq \tfrac{\lambda n \ell^*}{20\gamma} - \tfrac{1}{4}\lambda n 
\geq \tfrac{|\hOPT^*|}{40\gamma} -\tfrac{1}{4} \lambda n
\geq \tfrac{|\hOPT|}{40\gamma\log (4/\lambda)}- \tfrac{1}{4}\lambda n 
> \tfrac{|\OPT|}{40\gamma\log (4/\lambda)}- \lambda n.
\]

\paragraph{Sparse Case:}
Now, suppose $q^* < \tfrac{1}{\gamma}$. Then $|U^* \cap (g(*)_2)^{-1}(1)\cap (([n] \setminus M) \times [k])| < \tfrac{nk}{\gamma}$ 
and hence for any $\tau=\tau_s$ such that $T=T_{1/c}(g(*)_2)$, 
we have
\[
G^*_{2} \geq \tfrac{|\OPT^*|}{\alpha}-\lambda n \geq \tfrac{|\OPT|}{\alpha \log(4/\lambda)}-\lambda n,
\]
where
\begin{align*}
\alpha &= a_s\left(n,\tau,\tfrac{\lambda n}{|U^* \, \cap \, g(*)^{-1}(1) \, \cap \, (([n] \setminus M) \times [k])|}\right) \leq a_s\left(n,\tau,\tfrac{\lambda \gamma}{k}\right). \\
c &= c_s\left(n,\tau,\lambda,\tfrac{\lambda n}{|U^* \, \cap \, g(*)^{-1}(1) \, \cap \, (([n] \setminus M) \times [k])|}\right) \leq c_s\left(n,\tau,\lambda,\tfrac{\lambda \gamma}{k}\right). 
\end{align*}

Since we are either in the 
{\em dense} or {\em sparse} case as described above,
we obtain that 
\[
\max\{G^*_1,G^*_2\} \geq \tfrac{|\OPT|}{A\log(1/\lambda)} - \lambda n,
\]
where $A = \max\left\{O(\gamma),a_s\left(n,\tau,\tfrac{\lambda \gamma}{k}\right)\right\}$.

We conclude that $\widehat{G} = \max_\ell\{G_{\ell,1},G_{\ell,2}\}$ is a $(A \log(1/\lambda),\lambda n)$-approximation for $|\OPT|$, as needed.
\medskip

\paragraph{Runtime Complexity:}
The algorithm's runtime is dominated by the first step, i.e., generating increasing subsequences, which requires $\tilde{O}(nk)$ time using 
\Cref{thm::dynamic_main}. In addition, we require 
$t_s$ time for computing 
the second estimator. Thus, the overall algorithm for \estgenlis runs in time $\tilde{O}(nk) + t_s$.
\end{proof}

\section{Algorithm for \reslis: Proof of 
\Cref{thm::sl_rec_tree}}\label{sec::reslis}

We now present our construction for the $\reslis$ problem, and prove~\Cref{thm::sl_rec_tree}. 
Recall that we are given an interval $X \subseteq [n]$, a sequence 
$y \in \N^{X \times [k]}$
and a range of values 
$Y \subset \N$, and the goal is to approximate $\blocklis(y,X,Y)$, i.e., the length of a maximal sub-sequence $\OPT$, specified by 
a set of indices $w_1,w_2,\ldots,w_\ell \subseteq X \times [k]$ such that
   the set of (first coordinate, value) pairs $\{((w_i)_1,y_{w_i})\}_{i \in [\ell]}$ 
   is a subset of $X \times Y$, 
   and is a monotone set.

Our construction, at a high level, reduces the $\reslis$ problem above to a $\genlis$ instance over $\tau$ blocks, where $\tau$ is the branching factor of \Cref{thm::sl_rec_tree}, whose value is given as input and controls the delicate 
balance between approximation and complexity.  
The first coordinates of 
the $\genlis$ instance are 
computed directly, 
while each genuineness 
flag 
corresponds to a $\blocklis$ instance of a block interval $X_i$ of some smaller size $r \triangleq |X|/\tau$, and restricted to a $y$-interval $Y' \subseteq Y$. We consider the pair $X_i,Y'$ as genuine only if the (recursive) approximation of $\blocklis(y,X_i,Y')$ is 
above a certain {\em threshold} $\kappa$. 
For each $x$-interval $X_i$, we compute a set of $y$-intervals $\cY_i$ (called \emph{candidate intervals}), which together form
a \emph{block}
in the $\genlis$ instance. Intuitively, one can think of the $\genlis$ instance as solving the {\em global $\LIS$}, while each $\reslis$ instance is solving some {\em local $\LIS$} over (sub-)interval $X_i$, over some range of $y$-values. 
Overall, we will
show that such a formulation is equivalent to a composition of functions, and use the Tree Decomposition Lemma to obtain our correctness and complexity guarantees.

\subsection{Extending the 
\genlis problem}
Before describing our construction for proving~\Cref{thm::sl_rec_tree}, we  introduce a slightly stronger version of it, which will be easier to work with. In particular, we introduce 2 extensions of the requirements of 
the \genlis algorithm. 
First, we allow the input $y$-values to be {\em intervals} rather than integers, and second, we require the approximation to improve if there are many ``null'' elements,
i.e., if most blocks have much less than $k$ 
elements.\footnote{We highlight that ``null'' is different from
``not genuine''; while we need to test to determine that an element is ``not genuine'', 
``null'' elements can be determined ``for free''.} In particular, to account for the latter, we note that the input size is $y(X)\cap Y$, which can be much less than $|X|\cdot k$.

 Overall, the stronger version is the following:

\begin{lemma}[\Cref{thm::sl_rec_tree}, extended] \label{lm::sl_rec_relaxed}
Fix monotone 
functions $a_g :  
\R_+^2 \rightarrow [1,\infty)$, $a_s,c_g:\R_+^3 \rightarrow [1,\infty)$ and $c_s:\R_+^4 \rightarrow [1,\infty)$, satisfying, for all $r, \tau \in E_\beta$, $m \in \N$, $\lambda < 1$, $\lambda_1, \lambda_2 \in [\lambda,1]$, and $k' \in [1,1/\lambda]$ with $\lambda_1 \lambda_2 =
	 \Omega(\lambda k' 
	 )$: 

	\begin{itemize}
	\item
$a_s\left(r,\tau,\tfrac{\lambda}{m}\right) \geq \polylog \left(\frac{k}{\lambda}\right) \cdot a_g\left(\tau, \frac{\lambda_2}{k'}\right) \cdot a_s\left(\frac{r}{\tau},\tau,\frac{\lambda_1}{m}\right)$; and
\item
$c_s\left(r,\tau,\lambda,\tfrac{\lambda}{m}\right) \geq  \polylog (n) 
\cdot \frac{\tau}{\lambda} + c_g(\tau,\lambda_2,k') \cdot c_s\left(\frac{r}{\tau},\tau,\Theta(\lambda_1/k'),\frac{\lambda_1}{m}\right)$. 

	\end{itemize}
	
	Fix input $y \in \N^{n \times k}$, a block interval $X \subseteq [n]$, value range interval $Y \subseteq \N$, parameters $\lambda \in (0,1)$, $\beta \in \N$, $\tau \in E_\beta$.	
Suppose there exist the following algorithms:
\begin{enumerate}
\item[$\cA_{GL}$:] 
Given $\genlis$ instance $g \in (\N \times \zo)^{n_g \times k_g}$, with $m_g$ non-null elements total, 
and $\beta$-ary \prectree $T_{1/c_g(n_g,\lambda_g,m_g/n_g)}(g(*)_2)$ access, $\cA_{GL}$ outputs  $\left(a_g\left(n_g,\tfrac{\lambda_g n_g}{m_g}\right),\lambda_g n_g\right)$-approximation to $\genlis(g)$ in time 
$\tO(n_g + m_g)$ w.h.p.
\item[$\cA_{BL}$:] Fix $t>1$. For any interval $X' \subset X$, value range $Y' \subset Y$, and any $\lambda_s < 1$, 
given a $\beta$-ary \prectree $T_{1/c_s\left(|X'|,\tau,\lambda_s,\tfrac{\lambda_s |X'|}{|y_{-\lambda/2}(X',Y')|}\right)}(y(X'))$, $\cA_{BL}$ outputs $\left(\alpha_s,\lambda_s |X'|\right)$-approximation to $\blocklis(y,X',Y')$  w.h.p., where $\alpha_s \triangleq a_s\left(|X'|,\tau,\tfrac{\lambda_s |X'|}{|y_{-\lambda/2}(X',Y')|}\right)$. The expected run-time is at most $t_{BL}=t \cdot c_s\left(|X'|,\tau,\lambda_s,\tfrac{\lambda_s |X'|}{|y_{-\lambda/2}(X',Y')|}\right)\cdot \tfrac{ |y(X',Y')|}{ |X'|}$.
\end{enumerate}


	 Then, 
	 given a $\beta$-ary \prectree $T_{1/c}(y)$, 
	 we can produce a $(\alpha, \lambda |X|)$-approximation for $\blocklis(y,X,Y)$ w.h.p., as long as
$\alpha \geq a_s\left(|X|,\tau,\tfrac{\lambda |X|}{|y(X',Y')|}\right)$ and
$c \geq c_s\left(|X|,\tau,\lambda,\tfrac{\lambda |X|}{|y(X',Y')|}\right)$.

		The algorithm's expected run-time is at most $
	c
	\cdot t \cdot \log^{O(\log_\beta(\tau))}(\beta)  \cdot \tfrac{ |y(X,Y)|}{ |X|}$.
\end{lemma}

To prove \Cref{thm::sl_rec_tree}, we use the reduction from \Cref{lm::reslis_int_reduction}, along with the following reduction:

\begin{lemma}\label{lm::genlis_int_reduction}
    Suppose the algorithm $\estgenlis$ $(a(\lambda/k),\lambda n)$-approximates $\genlis$ for some function $a$ in time $t(n \cdot k)$ over integers, then there exists an algorithm $\cA'$ that $(O(a(\lambda n/m) \log(1/\lambda)),\lambda n)$-approximates $\genlis$ over intervals with $m$ non-null elements.
    The run-time is $t(m) + O(n + m)$.
\end{lemma}
 

The proof of \Cref{lm::genlis_int_reduction} is deferred to \Cref{sec:genlis_ext}.

\begin{proof}[Proof of \Cref{thm::sl_rec_tree} using \Cref{lm::sl_rec_relaxed,lm::genlis_int_reduction,lm::reslis_int_reduction}]
The proof follows by induction. We assume \Cref{thm::sl_rec_tree} holds for $X' \subset X$, setting $t = t_s(|X'|,\tau)$ (noting that for the base case, we have $a_s(1,\cdot,\cdot) = 1$ and $c_s(1,\cdot,\cdot,\cdot) = 1$, simply by outputting 1 if and only if the block is not empty). 
Also note that the extended algorithms assumed in \Cref{lm::sl_rec_relaxed} reduce to the standard version using \Cref{lm::genlis_int_reduction}, with $O(\log 1/\lambda_g)$ additional 
approximation (which can be absorbed into 
the $\polylog (k/\lambda)$-approximation of parameter $\alpha$). 
Finally, one can observe that 
the additive time of $\tO(n_g + m_g)$ does not change the asymptotic time complexity.
\end{proof}

The rest of this section is devoted to proving
\Cref{lm::sl_rec_relaxed}.

\subsection{Algorithm}
The algorithm for the \reslis problem, named \estreslis, is described in 
\Cref{fig::ALG_reslis_amended} and uses routines that we describe next.

We first provide some intuition underlying the algorithm. The algorithm is based on a reduction to a \genlis instance, where checking whether a item is genuine takes one \reslis call. We view the \genlis instance as the ``global LIS'' and the \reslis instances as ``local LIS''.

Consider the case when $k = 1$, in which case the problem becomes a standard $\LIS$ problem. 
In this case, one visualize the instance as a set of points on a standard two-dimensional grid, with $X$ being an interval on the $x$-axis, and the element $(i,y_i)$ being represented by the point with $x$-coordinate $i$ and $y$-coordinate $y_i$. In this equivalent formulation, the objective is to determine the maximum length of a subsequence that is increasing with respect to both axes.

Assume that the longest increasing subsequence $\OPT$ is of length approximately $\lambda |X|$. Note that $\OPT$ can be distributed arbitrarily with respect to $X$. We split $X$ into sub-intervals: 
for some length $r>1$, use the standard, in-order partitioning of $X$ into $\tau = |X|/r$ mutually disjoint and covering intervals $X_1, \cdots, X_\tau$ of length $r$ each. There are two potential
extreme scenarios for the distribution of elements from $\OPT$:

\begin{enumerate}
    \item 
    All elements in $\OPT$ belong to approximately $\lambda |X| / r$ intervals $X_i$, and each such interval consists entirely of elements in $\OPT$; 
    the other intervals do not contribute at all.
    \item $\OPT$ is uniformly distributed across all intervals, i.e., each interval $X_i$ contributes approximately $\lambda r$ elements to $\OPT$.
\end{enumerate}

Intuitively speaking, the problem of {\em certifying a ``close to optimal'' monotone sequence} is more difficult when fewer elements participate in an optimal sequence.
Therefore, in the first case above, one has to ``work harder'' on the global \LIS to find the 
participating intervals $X_i$; however, little effort is required to verify the local \LIS in each interval of interest, i.e., to get a lower bound on the \LIS within each interval. 
In contrast, in the second case, it should be easier to determine the global \LIS, but it is more difficult
to approximate the local \LIS (within each interval $X_i$).

Of course, the distribution of $\OPT$ among the intervals $X_i$ can be arbitrary, between these extreme scenarios. However, one can 
show that there must exist some $\rho \in [1,1/\lambda]$, such that there are approximately $1/\rho$ fraction of
 intervals $X_i$, each having approximately $\lambda \rho r$ ``local'' contribution to $\OPT$.

\subsubsection{Algorithm Overview}

Now we describe the algorithm for the general $\reslis$ problem.
Before describing the main algorithm, we 
describe its 
main subroutines and mechanism, 
setting up some concepts and notation along the way. The notation is summarized 
in \Cref{tb::notation}.
The algorithm assumes no element repeats itself in the sequence. To remove such assumption, one can reconfigure the element values setting $y_{i,j} \gets n \cdot y_{i,j} - i$.


    \paragraph{Sampling.} 
The algorithm starts by sampling uniformly random blocks
$(w_j, y_{w_j})$ and then accumulating all $y_{w_j}$-samples in $Y$ within 
the sampled blocks (we note that sampling is accomplished via the SODS over the input \prectree). The main idea is to generate a set $\cS_i$ of sampled values, such that the ``distance in $\cS_i$'' between any 2 values in $\cS_i$ is approximately proportional 
to the number of integers in the range over the entire $y(X_i)$. To obtain such a guarantee, we simply use members of $\cS_i$ whose rank is a multiple of $O(\log n)$. We also want to make sure there are not too many total samples in all blocks, i.e., that $\sum_i |\cS_i|$ is close to 
its expectation. 
To control the tail bounds of such a  quantity, we also discard the $\Theta(\log n)$ blocks with the largest number of samples in $Y$. 
The sampling algorithm is presented in \Cref{fig::ALG_reslis_sample}, and the guarantees are formally stated in \Cref{cl::sample_inv}, as part of the analysis.


\begin{figure}[h!] 
\setstretch{1.2}
  \begin{algorithm}[H]
    \caption{\textsc{SampleAndPartition}}\label{fig::ALG_reslis_sample}
		\DontPrintSemicolon
		\SetNoFillComment
		\KwIn{
A \prectree $T=T_{1/c}(y)$, where $y \in \N^{X \times [k]}$
consists of $|X|$ blocks with at most $k$ 
integers each, 
value range $Y \in \cI$,
branching parameter $\tau$, and sample size $s$.
		}	
		\KwOut{Partition of $X$ with sampled $y$-values of each part.
		}
	 $\zeta \gets O(\log n)$.\\ 
	          $X_1,\ldots,X_{\tau} \gets$ a partition of $X$ into $\tau$ 
     consecutive intervals, each consisting of 
     $r \triangleq |X|/\tau$ 
     blocks.\\
    Using \Cref{cr::sampling_ds} and $T$, 
    sub-sample blocks $S \subseteq X$
    i.i.d. with probability $\zeta\tfrac{s}{|X|}$, and let 
    $W \gets \{(i,j) \in S \times [k] \mid y_{i,j} \in Y\}$. \\
    Let $B \subseteq S$ be the $\zeta$ largest
    blocks in $S$, i.e., for which
    $|B\times[k] \cap W|$ is maximized.\\

     \For {$i \in [\tau]$}       
        	{
     $\{p^{i}_1,p^{i}_2,\ldots\} \gets 
     y\left(W \cap ((X_i \setminus B) \times [k])\right)$ 
     where $p^i_m \leq p^i_{m+1}$ for all $m$.\\

     $\cS_i \gets 
    \{p^i_m\}_{m \in [|p^i_*|] \cap \zeta \N}$. \\
    
     }
    \Return $(X_*,\cS_*)$.
\end{algorithm}
\caption{Description of the sampling and partitioning subroutine.}
\end{figure}

       \paragraph{Generation of Candidate Intervals.} 
         Now we describe our construction of the candidate intervals $\cY_{i}$ using the sampled $y$-values $\cS_i$, for some fixed $i$. The high-level goal, as in \cite{rubinstein2019approximation}, is to cover an optimal \LIS solution using a small number of monotone interval sequences denoted as \emph{pseudo-solutions}, so that the largest \LIS within all pseudo-solutions would be a good approximation to the optimum. However, we need to choose such candidates more carefully, since, in addition to ``capturing the local $\LIS$", we would like to ensure the following efficiency guarantees, for our delicate bounds:

\begin{enumerate}
    \item There are not too many candidate intervals, i.e., $|\cY_i| \lessapprox \tfrac{|y(X_i,Y)|}{\lambda |X|}$; and more importantly,
    \item The total number of integers in range over all candidate intervals do not cause 
    significant 
    overhead. In particular: $\sum_{Y' \in \cY_i} |y(X_i,Y')| \lessapprox |y(X_i,Y)|$.
\end{enumerate}

In particular, we argue that given a value set $\cS$ (corresponding to the sampled values in each $X_i$) and a sub-additive set function $\psi$, we can generate a near-linear set of candidate intervals which approximates, up to a logarithmic factor, all candidate intervals. Here we note that the set function $\psi_{X_i,y}(Y') \triangleq \blocklis(X_i,Y',y)$ is sub-additive in $Y'$. While we did not manage to create a single set of intervals $\cY_i$ as above, we are able to create a {\em small family of sets} $\{\cY_{i,\Delta}\}_{\Delta \in E_2(|\cS_i|)}$ with similar guarantees. We show the following:

\begin{lemma}\label{lm::int_cluster}
	Fix a finite set of values $\cS \subseteq \N$ 
	of size $m$. There exist sets of intervals $\{\mathfrak{C}^\cS_\Delta\}_{\Delta \in E_2(m)} \subseteq \cI$ such that:
	\begin{enumerate}
		\item For all $\Delta$ and all $I \in \mathfrak{C}^\cS_\Delta$: $|I \cap \cS| = \Delta$.
		\item For all $\Delta$: $|\mathfrak{C}^\cS_\Delta| = |\cS|/\Delta$.
		\item Fix a sub-additive set function $\psi: \cI \rightarrow \R$. For any $I \in \{[a,b] \in \cI \mid a,b \in \cS\}$, there exists some $J \in \cup_\Delta \mathfrak{C}^\cS_\Delta$ with $J \subset I$ 
		such that $\psi(I) = O(\log |\cS|) \cdot \psi(J)$.
	\end{enumerate}
\end{lemma}

%

We prove \Cref{lm::int_cluster} in \Cref{sec::int_dyadic}.
           For now, for each
     set $\cS_i=\{p_1^i,p_2^i,
     \ldots\}$, 
    we construct, on average, $O(\log(k/\lambda))$ \emph{candidate interval sets} $\{\cY_{i,\Delta}\}_\Delta$ where $\Delta\in E_2(|\cS_i|)$ and each $\cY_{i,\Delta}$ contains precisely $\Delta$ (consecutive) elements of $\cS_i$. In particular, $\cY_{i,\Delta}$ are the dyadic intervals of $\cS_i$, i.e., $\cY_{i,\Delta}$ is the set of all the intervals $[p^i_{k\Delta+1},p^i_{(k+1)\Delta}]$ for all possible integer $k\ge 0$.

\vspace{2mm}

            \paragraph{Discretization.} The algorithm performs exponential discretization over the following parameters:
\begin{enumerate}
	\item [$\rho$:] 
	inverse of the fraction of intervals participating in $\OPT$. This parameter characterizes the relation between the ``global'' and ``local'' \LIS. In particular, for the two extreme scenarios outlined above, the parameter $\rho$ would be $1/\lambda$ and $1$ respectively.
	\item [$\Delta$:] a quantity 
	proportional to the ``candidate interval size'' of the local-$\LIS$ (namely, to $|y_{-\lambda/2}(y,X_i,Y')|$) and inversely proportional to the ``average block size'' of the global-$\LIS$ problem (namely, to $m_g/n_g$\footnote{Recall from \Cref{lm::sl_rec_relaxed} that $m_g$ is the total number of non-null elements of the instantiated \genlis problem and $n_g$ is the number of blocks.}).
\end{enumerate}

\paragraph{Decomposition and Recursion.} The next step is to determine, for each choice of parameters 
     $(\rho,\Delta)$, the largest set of monotone candidate intervals $\{(i,Y'): Y' \in \cY_{i,\Delta}\}_i$ such that each candidate interval has local $\LIS$ at least $\kappa_\rho \triangleq r/\rho$. 
     For this purpose, 
     the algorithm defines the pair $(i,Y')$ as ``genuine'' iff the local LIS is long enough (determined via recursive call to \reslis). Then the problem becomes to solve a ``global" $\genlis$ instance over all genuine $(X_i, Y')$ pairs. 
     In other words, we formulate the problem as a composition of a $\genlis$ problem over many smaller $\reslis$ problems, where the \prectree decomposition algorithm (\Cref{lm::tree_decomposition}) is used to access the sequence question.

    \paragraph{Optimization.} Finally, 
    we output the solution with the maximum value among all 
    combinations of the parameter $\rho$ and $\Delta$.
    
    \begin{figure}[h!] 
\setstretch{1.2}
  \begin{algorithm}[H]
    \caption{\textsc{Decompose}}\label{fig::ALG_reslis_decompose}
		\DontPrintSemicolon
		\SetNoFillComment
		\KwIn{
		A \prectree $T=T_{1/c}(y)$, where $y \in \N^{X \times [k]}$ consists of $|X|$ blocks with at most $k$ 
		integers each,
		$x$-partition $\{X_i\}_i$ of $X$, $y$-partition sets $\{\cY_{i}\}_{i}$ (for a fixed $\Delta$) of some interval $Y \in \cI$.
			}	
		\KwOut{An integer $\widehat{L} \in \left[0, \reslis(y,X,Y) \right]$.
		}
		Let $\agl$, $\abl$ be the aforementioned algorithms from \Cref{lm::sl_rec_relaxed}. \\
		Define $\kappa_\rho \triangleq \tfrac{r}{\rho}$.\\
        For each $i$, 
        define $\Lambda_{\rho}(i,*) \triangleq 
        	\left[\left(Y', \1\left[\abl\left(y(X_i),Y',\tfrac{1}{\rho},\tau\right) > \kappa_\rho\right]\right)\right]_{Y' \in \cY_{i}}$. \\
       \For{$\rho \in E_2(1/\lambda)$}{
       Compute $G_\rho \gets \agl(\Lambda_\rho, \lambda \rho) \cdot 0.5 \kappa_\rho$ 
       using \Cref{lm::tree_decomposition} and $T$.
       }
    \Return $\widehat{L} \gets \max_{\rho} G_{\rho}$.
\end{algorithm}
\caption{Description of the decomposition subroutine. Note that lines $2-4$ are merely definitions for $\rho \in E_2(1/\lambda)$, and not computational steps.}
\end{figure}

\begin{table}
\begin{center}
 \begin{tabular}{| l | l |} 
 \hline
 {\bf Symbol, Definition} 
 & {\bf Description} \\ 
 \hline
 \hline
 $\lambda$ & additive error parameter (as fraction); given as input\\
 \hline
 $\tau$ & branching factor; given as input \\ 
 \hline
 $X_i$ & sub-interval of $X$\\
 \hline
 $r := |X|/\tau$ & \emph{radius}, i.e., number of blocks in each interval $X_i$\\ 
 \hline
 $\rho$ & inverse of fraction of intervals participating in {\em global} $\OPT$ \\ 
 \hline
  $\kappa_\rho := r/\rho$ & the local-$\LIS$ threshold that qualifies a candidate interval as ``genuine''\\ 
 \hline
 $Y$ & permissible range of $y$-values\\
 \hline
 $\cS_i$ & $y$-values in $Y$ found in sampled blocks in $X_i$\\
 \hline
 $\cY_{i,\Delta}$ & candidate intervals for $X_i$ using samples of $\cS_i$ of ``distance'' $\approx \Delta$\\
 \hline
  $\ell_{i,Y'} := \blocklis(y,X_i,Y')$ & the length of the longest sub-sequence of $y(X_i)$ using values in $Y'$\\
 \hline
 $T := T_{1/c}(y)$ & $\beta$-ary \prectree with starting complexity/precision $c$\\
  \hline
 $a_g(\cdot,\cdot),c_g(\cdot,\cdot)$ &  approximation and complexity functions for \genlis\\
   \hline
 $a_s(\cdot,\cdot,\cdot),c_s(\cdot,\cdot,\cdot)$ & approximation and complexity functions for \reslis\\
\hline
$y_{-\delta, Y'}$
&
\pbox{20cm}{\vspace{1mm}block sequence $y$,  where the heaviest $\delta$ fraction of blocks \\
(considering only values in $Y'$) are emptied}
	\\
  \hline
  $y_{-\delta}(X',Y') := y(X')_{-\delta,Y'} \cap Y'$ & multiset of values in $y(X')_{-\delta,Y'}$ restricted to $Y'$ \\
 \hline
\end{tabular}
\caption{Notation Table}
\label{tb::notation}
\end{center}
\end{table}

\ignore{
\begin{table}
\begin{center}
 \begin{tabular}{| l | l |} 
 \hline
 {\bf Symbol} & {\bf Description} \\ 
 \hline\hline
 $\tau$ & branching factor. Given as input to control balance of approximation / complexity.\\
 \hline
 $r$ & $=|X|/\tau$ \emph{radius}, i.e., number of blocks in each interval $X_i$.\\ 
 \hline
 $\rho$ & inverse of fraction of intervals participating in $\OPT$ \\ 
 \hline
  $\kappa_\rho$ & $=r/\rho$. the local-$\LIS$ threshold for making a candidate interval ``genuine''.\\ 
 \hline
 $\cS_i$ & $y$-values in $Y$ found in sampled blocks in $X_i$\\
 \hline
 $\cY_{i,\Delta}$ & candidate intervals generated over $X_i$ using samples $\cS_i$ of ``sample-length'' $\approx \Delta$\\
 \hline
\end{tabular}
\caption{Notation Table}
\label{tb::notation}
\end{center}
\end{table}
}

\begin{figure}[h!] 
\setstretch{1.2}
  \begin{algorithm}[H]
    \caption{\estreslis}\label{fig::ALG_reslis_amended}
		\DontPrintSemicolon
		\SetNoFillComment
		\KwIn{
		A \prectree $T=T_{1/c}(y)$, where $y \in \N^{X \times [k]}$ consists of $|X|$ blocks with at most $k$ 
		integers each,
		 value range $Y \subseteq \N$,
		 error parameter $\lambda \in (1/n,1)$, branching parameter $\tau$.
			}	
		\KwOut{An integer $\widehat{L} \in \left[0, \reslis(y,X,Y) \right]$.
		}

     \If{$|X| < c$}{\Return $\widehat{L} \gets \blocklis(y,X,Y)$ through exact computation.}
     
    
	$(X_*,\cS_*) \gets \textsc{SampleAndPartition}(T,\tau,\tfrac{\tau}{\lambda})$. \\
	
        	{
     $\cY_{i,*} \gets$ sets of intervals by invoking \Cref{lm::int_cluster} using each $\cS_i \in \cS_*$.\\
     
      \Return $\widehat{L} \gets \max_\Delta \{\textsc{Decompose}(T,X_*,\cY_{*,\Delta})\}$.

       }
\end{algorithm}
\caption{Description of the algorithm $\estreslis$.}
\end{figure}

\subsection{Main analysis}\label{sec::reslis_analysis}

Before analyzing the estimator $\widehat{L}$, we show several important properties. First, we claim that partitioning 
$X$ into intervals of smaller size 
and matching each one with a value range $Y$ monotonically cannot over-estimate the overall $\LIS$:

\begin{lemma}[Upper Bound]\label{lm::blis_upper}
	Fix an interval $Y \subset \N$. Let $X_1,\ldots,X_\tau$ be a partition of $X$ into consecutive, disjoint, in-order block intervals (i.e., $X_1 < X_2 < \ldots < X_\tau$) and let $P \subset [\tau] \times \cI$ 
	be an arbitrary {\em monotone} set. Then, we have $\sum_{(i,Y_i) \in P} \blocklis(y,X_i,Y_i \cap Y) \leq \blocklis(y,X,Y)$.
\end{lemma}
\begin{proof}
For each $(i,Y_i) \in P$, let $\OPT_i \subseteq (X_i \times [k]) \, \cap \, y^{-1}(Y_i \cap Y)$ be a set of coordinates of an optimal increasing subsequence for $\blocklis(y,X_i,Y_i \cap Y)$ 
Then $\sum_{(i,Y_i) \in P} |\OPT_i| 
= \sum_{(i,Y_i) \in P} \blocklis(y,X_i,Y_i \cap Y)$. On the other hand, since  
      $P$ is a monotone set, we have $y(p) < y(q)$ for any $p \in \OPT_i, q \in \OPT_j$ with $i < j$. We conclude that 
the sequence $\cup_i \, y(\OPT_i) \subseteq Y$ forms an increasing subsequence of $y(X)$ of length $\sum_{(i,Y_i) \in P} \blocklis(y,X_i,Y_i \cap Y)$ 
and the claim follows.
\end{proof}

Next, we argue three essential invariants when constructing candidate intervals. These guarantees will be used later to lower bound our estimator:

\begin{lemma}[Candidate intervals guarantees]\label{lm::reslis_invariants_tree}
    Fix input $X,Y, y$. With high probability, 
    all the following invariants hold for the $\estreslis$ algorithm:    \begin{enumerate}
	 
    	  \item  For all $i,\Delta$ and all $Y' \in \cY_{i,\Delta}$,
       	\[
    	    	|y_{-\lambda/2}(X_i,Y')| = O\left(\lambda \Delta r\right).
    	\]
    	\item For all $\Delta$,
    	\[
    	\E_{i \in [\tau]}\left[|\cY_{i,\Delta}|\right] = O\left(\frac{|y(X,Y)|}{\lambda \Delta |X|}\right)
    	,
    	\]
		\item There exist candidate intervals $Y' \in \oplus_{i \in [\tau]} \cup_{\Delta \in E_2(\poly(k/\lambda))} \cY_{i,\Delta}$  
		such that 
		    $\{(i,Y'_{i})\}_i$ is a monotone set; and 
    		\[
		\sum_i 
		    \blocklis(y,X_i,Y'_i) = 
		    \Omega \left(\frac{\blocklis(y,X,Y) - \lambda |X|}{\log (k/\lambda)}\right).
    		\]
    \end{enumerate}
\end{lemma}

%
%
%
%


We analyze the correctness and complexity of our algorithm, proving \Cref{lm::sl_rec_relaxed} assuming \Cref{lm::reslis_invariants_tree}. 
We prove \Cref{lm::reslis_invariants_tree} in \Cref{sec::resinv_proof}. 




	

\begin{proof}[Proof of \Cref{lm::sl_rec_relaxed} using \Cref{lm::reslis_invariants_tree}]

	For all $i,\Delta$ and all $Y' \in \cY_{i,\Delta}$, let $\ell_{i,Y'} \triangleq \blocklis(y,X_i,Y')$. 
	By \Cref{lm::reslis_invariants_tree} 
	(3), we have that for some $\Delta^*$, 
	there exists a monotone set of pairs 
	$P^* \in \oplus_i (\{i\} \times \cY_{i,\Delta^*})$
	such that:
	$$\sum_{(i,Y') \in P^*} \ell_{i,Y'} \geq \Omega\left(\tfrac{\blocklis(y,X,Y) - \lambda |X|}{\log (k/\lambda)\cdot \log(k/\lambda)} \right) 
	> \tfrac{\Omega(\blocklis(y,X,Y)) - \lambda |X|}{\log^2 (k/\lambda)}.$$
	
	
	
	Therefore, we have that, for any $\alpha_s > 1$, there is some $\rho^* \in E_2(1/\lambda)$ for which, using $\kappa_{\rho^*} = r/\rho^*$:
	\begin{equation}
	\sum_{(i,Y') \in P^*} \1\left[\ell_{i,Y'} \geq 
	2\alpha_s \cdot \kappa_{\rho^*} \right] \geq \tfrac{1}{2\alpha_s \kappa_{\rho^*}} \cdot \tfrac{\Omega(\blocklis(y,X,Y)) - \lambda|X|}{\log^3 (k/\lambda)}
		\end{equation}
		
    Let us first argue correctness. Define 
     $\kappa^* \triangleq \kappa_{\rho^*} = \tfrac{r}{\rho^*}$ and $\cY_i^* \triangleq \cY_{i,\Delta^*}$. 
%
Recall the following 
 guarantees of $\abl$ and $\agl$ which hold for each pair $(i,Y') \in P^*$:
    \begin{enumerate}
        \item Let $\widehat{\ell_{i,Y'}}$ be the output of $\abl(y(X_i),Y',\tfrac{1}{\rho^*},\tau)$. 
        Then $\widehat{\ell_{i,Y'}}$ is a 
        $\left(\alpha_s,\tfrac{r}{\rho^*}\right)$
        -approximation of $\ell_{i,Y'}$, where $\alpha_s \triangleq a_s\left(r, \tau, \tfrac{r/\rho^*}{|y_{-\lambda/2}(X_i,Y')|}\right)$. 
        
        \item Let $\Lambda^* \triangleq \Lambda_{\Delta^*,\rho^*}$. The output of $\agl(\Lambda^*, \lambda \rho^*)$ is a $\left(\alpha_g 
        , \lambda \rho^* \tau \right)$-approximation     
        of $\sum_{(i,Y') \in P^*} \1\left[\widehat{\ell^*_{i,Y'}} \geq \kappa^*\right]$, where $\alpha_g \triangleq a_g\left(\tau,\tfrac{\lambda \rho^* n_g}{ m_g}\right)$, where $n_g = \tau$ and $m_g = \sum_i\left[|\cY_{i}^*|\right]$, are, respectively,  the number of blocks and non-null elements in the \genlis instance.
    \end{enumerate}
     We use the approximation guarantees from above to obtain a lower bound on $ L^* \triangleq L_{\Delta^*,\rho^*}$. Assuming $c \geq c_s(\cdot) \cdot c_g(\cdot)$, 
     we invoke \Cref{lm::tree_decomposition} to get that the quantity $L^*$ is a $(\alpha_g, 
  \lambda \rho^* \tau)$-approximation for a \genlis instance, where each block $i$ is defined by all possible candidate intervals $\cY^*_i$, and each candidate interval $Y' \in \cY^*_i$ is genuine iff a    
     $\left(\alpha_s,\tfrac{r}{\rho^*}\right)$-local 
   approximation $\widehat{\ell_{i,Y'}}$ passes the $\kappa_{\rho^*}$ threshold. 
    
Note that the approximations are:
\begin{align*}
      	\alpha_s &= a_s\left(r, \tau, \tfrac{r}{\rho^*|y_{-\lambda/2}(X_i,Y')|}\right) 
      	\\
      	\alpha_g &= a_g\left(\tau,\tfrac{\lambda \rho^*}{ \E_i\left[|\cY_{i}^*|\right]}\right) 
      	.
      \end{align*} 
    We let $\lambda_s=\tfrac{|y(X,Y)|}{|X|}\cdot \tfrac{r}{\rho^*|y_{-\lambda/2}(X_i,Y')|}$, and $\lambda_g=\tfrac{m_g}{n_g}\cdot \tfrac{\lambda \rho^*}{ \E_i\left[|\cY_{i}^*|\right]}=\lambda \rho^*$, in which case:
    \begin{equation}
    	     \alpha_s \cdot \alpha_g = a_s\left(r,\tau,\tfrac{\lambda_s|X|}{|y(X,Y)|}\right) \cdot a_g\left(\tau,\tfrac{\lambda_g n_g}{m_g}\right).
    \end{equation}

 We need to bound $\lambda_s\lambda_g$. For this, we use 
         \Cref{lm::reslis_invariants_tree} 
         (1) and (2), to obtain for each $Y' \in \cY^*_{i}$, for all $i$:
     \begin{equation}
		|y_{-\lambda/2}(X_i) \cap Y'| = O\left(\lambda \Delta r\right); 
     \end{equation}     
     \begin{equation}\label{eq::avg_k}
		\tfrac{m_g}{n_g}=\E_i\left[|\cY_{i}^*|\right] = O\left(\tfrac{|y(X,Y)|}{\lambda \Delta |X|}\right).  
     \end{equation}
    
We obtain that 
$$
\lambda_s\lambda_g\cdot n_g/m_g=\Omega\left(\tfrac{|y(X,Y)|}{|X|}\cdot \tfrac{r}{\rho^*\lambda\Delta r}\cdot \lambda\rho^*\cdot \tfrac{\lambda \Delta |X|}{|y(X,Y)|}\right)=\Omega(\lambda),
$$
as required in the theorem statement.
	
     We can finally derive the following overall bound on $\widehat{L}$, and in particular, show it is a $(\alpha,\lambda |X|)$-approximation, for some $\alpha = O\left(\alpha_s \cdot \alpha_g \cdot \polylog(k/\lambda)\right)$. For the lower bound, recall $\widehat{L} = \max_{\rho,\Delta} L_{\rho,\Delta}\geq L^*$, 
     and therefore we have that:
     
    \begin{align*}
        \widehat{L} &\geq L^* \\
        &\geq \left(\tfrac{1}{\alpha_g}\sum_{(i,Y') \in P^*} \1\left[\widehat{\ell_{i,Y'}} \geq \kappa^* \right] - \lambda \rho^* \tau  \right) \cdot 0.5 \kappa^* 
        \\
        &\geq \sum_{(i,Y') \in P^*} \tfrac{1}{\alpha_g}\1\left[\ell_{i,Y'} \geq \alpha_s \cdot \left(\kappa^* + \tfrac{r}{\rho^*}\right)\right]  \cdot 0.5 \kappa^* - \lambda \rho^* \tau\cdot 0.5 \kappa^*\\
        &\geq \tfrac{0.5 \kappa^*}{\alpha_g} \cdot \sum_{(i,Y') \in P^*} \1\left[\ell_{i,Y'} \geq \alpha_s \cdot \left(\kappa^* + \tfrac{r}{\rho^*}\right)\right] - 0.5  \lambda |X| \\
        &\geq \tfrac{0.5 \kappa^*}{\alpha_g} \cdot \sum_{(i,Y') \in P^*} \left(\1\left[\ell_{i,Y'} \geq \alpha_s \cdot \left(\kappa^* + \kappa^*\right)\right]\right) - 0.5 \lambda |X| \\
        &\geq \tfrac{0.5 \kappa^*}{\alpha_g} \cdot \tfrac{1}{2\alpha_s \kappa_{\rho^*}} \cdot \tfrac{\Omega(\blocklis(y,X,Y)) - \lambda|X|}{\log^3 (k/\lambda)}
 - 0.5 \lambda |X| \\
    &> \Omega\left(\tfrac{\blocklis(y,X,Y)}{\alpha_s \alpha_g \log^3 (k/\lambda)}\right) - \lambda |X|. 
    \end{align*}
    
%

    For the upper bound, we use \Cref{lm::blis_upper} to claim that for any set of parameters, and any choice of monotone pairs, $L_{\rho,\Delta}$ is upper bounded by $\blocklis(y,X,Y)$, and hence any approximation of such quantity is guaranteed to be bounded by $\blocklis(y,X,Y)$ as well.

    \paragraph{Proving the assumed bound for $c$.}
For correctness, it is left to show the lower bound for $c$ suffices to obtain $c \geq c_{sampling}+
c_g(\cdot) \cdot c_s(\cdot)$ for every set of potential parameters from above, and where $c_{sampling}$ is the complexity for sampling the set $S_*$. Note that $c_{sampling}=\zeta s=\zeta \tau/\lambda$. As regards $c_s, c_g$, in order to obtain the desired approximations $\alpha_g,\alpha_s$, we need complexities $c_g=c_g(\tau,\lambda_g,m_g/n_g)$ and $c_s=c_s(r,\tau,1/\rho^*)$ respectively. \aanote{need to make sure of constants here too} Now, noting $\tfrac{m_g}{n_g} \stackrel{Eq.~(\ref{eq::avg_k})}{=} O\left(\tfrac{|y(X,Y)|}{\lambda \Delta |X|}\right)$, we have 
$$
\tfrac{1}{\rho^*} = \tfrac{\lambda \cdot \lambda_s \cdot \Delta \cdot |X|}{|y(X,Y)|} = \Omega\left(\tfrac{\lambda_s n_g}{m_g}\right),
$$

as needed.
         This concludes the correctness proof. 
               
    \paragraph{Runtime Analysis.}
    The algorithm consists of the following procedures:
    
    First, {\bf sub-sampling} of blocks, with expected runtime $\tO(\tfrac{\tau}{\lambda}) \cdot \tfrac{|y(X,Y)|}{|X|}$ by \Cref{cr::sampling_ds}.
    
    Second, {\bf computation of candidate intervals}, takes time $\tO(|S|) = \tO(\tfrac{\tau}{\lambda}) \cdot \tfrac{|y(X,Y)|}{|X|}$.
    
    Last, and most important, are the recursive calls invoked by the {\bf decomposition procedure}. Here, for each set of parameters $\Delta, t$, the expected runtime by \Cref{lm::tree_decomposition} is 
    \[
    \tO\left(\tau + t_{GL} + \tfrac{1}{\delta \tau} \cdot \norm{t_f}_1\cdot \log^{O(\log_\beta(\tau))}(\beta) \right),
    \]
    
    where $\delta = 
    \tfrac{1}{c_g(\cdot)}$. 
    The Lemma assumption provides the bound 
    $$
    t_{GL} = \tO(m_g + n_g) = \tO(\tau + \sum_i |\cY_{i,\Delta}|) = \tO(\tau + |S|) = \tO(\tfrac{\tau}{\lambda}) \cdot \tfrac{|y(X,Y)|}{|X|},
    $$
    
    and hence the main quantity to analyze is $\tfrac{1}{\delta \tau} \cdot \norm{t_f}_1\cdot \log^{O(\log_\beta(\tau))}(\beta)$.
    
    Notice that computing each $f_i$ requires 
    $|\cY_{i,\Delta}|$ invocations of $\abl$, where for each $Y' \in \cY_{i,\Delta}$, we pay (expected) time $t \cdot c_s(\cdot) \cdot \tfrac{|y(X_i,Y')|}{|X_i|}$. Furthermore, note that candidate intervals in $\cY_{i,\Delta}$ are disjoint and contained in $Y$, and hence $\sum_{Y' \in \cY_{i,\Delta}} |y(X_i,Y')| \leq |y(X_i,Y)|$. Therefore:
    

    $$\E_T\left[\norm{t_f}_1\right]
    \leq 
    \sum_i \sum_{Y' \in \cY_{i,\Delta}} t \cdot c_s \cdot \tfrac{|y(X_i,Y')|}{|X_i|}
    \leq
    \sum_i t \cdot c_s \cdot \tfrac{|y(X_i,Y)|}{|X_i|}
    = t \cdot c_s \cdot \tfrac{|y(X,Y)| \tau}{|X|}. 
    $$
    
    Now, since $c \geq c_s \cdot c_g$, then, 
    
    $$\E_T\left[\tfrac{1}{\delta \tau} \cdot \norm{t_f}_1\right]
    =
    \tfrac{c_g}{\tau} \cdot \E[ \norm{t_f}_1]
    \leq
    t \cdot c_s \cdot c_g \cdot \tfrac{|y(X,Y)|}{|X|}
        \leq
    t \cdot c \cdot \tfrac{|y(X,Y)|}{|X|}.
    $$

    Counting the time taken by all procedures above, the overall expected runtime 
     is:
    $$\tO\left(\tfrac{\tau}{\lambda}
    + t \cdot c \cdot \log^{O(\log_\beta(\tau))}(\beta)\right) \cdot \tfrac{|y(X,Y)|}{|X|} = t \cdot c \cdot \log^{O(\log_\beta(\tau))}(\beta) \cdot \tfrac{|y(X,Y)|}{|X|},
    $$
    
    as needed.
    
%
    

\end{proof}

\subsection{Dyadic Construction for Candidate Intervals: Proof of \Cref{lm::int_cluster}}\label{sec::int_dyadic}

We now describe the algorithm for generating candidate intervals, proving \Cref{lm::int_cluster}. First, we show how to generate a small ``covering'' set of candidate intervals over some fixed set $[m]$.

\begin{claim}\label{cl::int_cluster}
	 For any $m \in \N$, there exists 
	 a clustering of $[m]$ into intervals $\{\cJ_{\Delta,1},\cJ_{\Delta,2},\ldots,\cJ_{\Delta,m/\Delta}\}_{\Delta \in E_2(m)}$ where each $\cJ_{\Delta,j} \subseteq \cI \cap 2^{[m]}$ 
	 and further:
	\begin{enumerate}
		\item $|\cJ_{\Delta,j}| = \Delta$ for all $\Delta,j$. 
		\item For any interval $I \in [m] \cap \cI$, there exists a set of intervals $\mathfrak{C}_I \subseteq \{\cJ_{\Delta,j}\}_{\Delta,j}$ of size $|\mathfrak{C}_I| = O(\log m)$, such that $I$ is precisely covered by $\mathfrak{C}_I$; i.e., $I = \cup_{\cJ \in \mathfrak{C}_I} \cJ$.
	\end{enumerate}
\end{claim}
\newcommand{\rightnode}[0]{\textrm{right}}
\begin{proof}
	Consider a binary tree $T_m$ with leaves indexed by integers in $[m]$ and internal nodes indexed by the union of its children. The set of clusters corresponding to the indexing of all nodes in each level of the tree provides the above guarantees. 
	For (1), 
	we have a node at each height $h \in [\log_2 m]$ of length $\Delta=2^h$ hence that bound is immediate.
	
	For (2), we reconstruct $I$ as follows: let $a \triangleq \min\{I\}$ and $b \triangleq \max\{I\}$. We start with $v_a$ 
	(the leaf with index $a$) and add it to $\mathfrak{C}_I$. Now we go up the tree level by level, each time adding the {\em sibling} of the current node if and only if 
	it is contained in $I$. Similarly, we do the same process with $v_b$. 
	We note that this process generates $\mathfrak{C}_I \subseteq \{\cJ_{\Delta, j}\}_{\Delta,j}$ 
	of size $O(\log m)$. It is left to show that $I = \cup_{\cJ \in \mathfrak{C}_I} \cJ$. 
	
	For this task, consider any $c \in [m]$. One one hand, if $c \notin I$, no ancestor of $c$ will be added to $\mathfrak{C}_I$ by construction. On the other hand, for $c \in I$, 
let $v$ be the ``lowest'' node such that 
	either $a,c$ or $b,c$ are common descendants of $v$. Assume w.l.o.g. that $a,c$ are the common descendants (otherwise consider the symmetric argument). Now, if $c = a$, then we add $v_a = v_c$ to $\mathfrak{C}_I$. 
	Otherwise, $c \in (a,b)$, then $c \in \rightnode(v) \subseteq I$, 
	and hence $\rightnode(v)$ will be added to $\mathfrak{C}_I$. We conclude that $c \in I$ 
	if and only if $c \in \cup_{\cJ \in \mathfrak{C}_I} \cJ$ as needed.
	
	
\end{proof}

Second, We prove \Cref{lm::int_cluster}, by simple ordering of the set $\cS$, and reducing any interval over $[m]$ to an interval over $\cS$. 

\begin{proof}[Proof of \Cref{lm::int_cluster}]
	Index all values by rank such that 
	$\cS_1 \leq \cS_2 \leq \ldots \leq \cS_m$. 
	We now construct a clustering of $[m]$ as per \Cref{cl::int_cluster}, and output $\mathfrak{C}^\cS_\Delta \triangleq \{[\cS_a,\cS_b] \mid [a,b] \in \{\cJ_{\Delta,i}\}_i\}$. Then all the above guarantees are immediate (the last guarantee is obtained by sub-additivity).
\end{proof}

\subsection{Analysis of \reslis invariants: Proof of \Cref{lm::reslis_invariants_tree}}\label{sec::resinv_proof}



To prove \Cref{lm::reslis_invariants_tree}, we state the following claim:

\begin{claim}\label{cl::sample_inv}
	The algorithm \textsc{SampleAndPartition} generates $\{(X_i, \cS_i)\}_i$ pairs which, with high probability, satisfy the following for all $i \in [\tau]$ and $Y' \in \cI$: 
	\begin{enumerate}
		\item $\sum_{w \in X_i} \1[y(w) \cap Y' \neq \emptyset] \leq \lambda r \cdot (|Y' \cap \cS_i| + 2)$
		\item $|y_{-\lambda/2}(X_i,Y')|  \leq 4\lambda r \cdot (|Y' \cap \cS_i| +2)$.
	\end{enumerate}
\end{claim}

We now prove \Cref{lm::reslis_invariants_tree} using \Cref{cl::sample_inv}.

\begin{proof}[Proof of~\Cref{lm::reslis_invariants_tree}]
	Fix $i, \Delta$. By \Cref{lm::int_cluster} (1), we have for all $Y' \in \cY_{i,\Delta}$ that $|Y' \cap \cS_i| = \Delta + 1$, and we use \Cref{cl::sample_inv} (2), to obtain the first invariant.
	
	For the second invariant, we bound $\E_{i \in [\tau]}\left[|\cY_{i,\Delta}|\right]$ by arguing about 
	the random sets $W,B,\cS_*$ computed in the \textsc{SampleAndPartition} algorithm. Let $\cS = \cup_i \cS_i$, we would like to show high probability bound on $|\cS|$. Note that $|\cS|\cdot \zeta =\Theta(|W \setminus (B \times [k])|)$. Now, each block is 
	subsampled independently with probability 
	$\left(\tfrac{\zeta \, \tau}{\lambda |X|}\right)$. Then, $\E[|W|] = 
	\tfrac{\zeta \, \tau \, |y(X,Y)|}{\lambda |X|}$. For our task, we invoke \Cref{ft::stat_max}, 
 which implies $|S| = O\left(\tfrac{\tau |y(X,Y)|}{\lambda |X|}\right)$ w.h.p. 
	We also invoke \Cref{lm::int_cluster} (2) to obtain $|\cY_{i,\Delta}| = \tfrac{|S_i|}{\Delta}$, which implies 
	\[\E_i\left[|\cY_{i,\Delta}|\right] = \E_i\left[\tfrac{|S_i|}{\Delta}\right] = \tfrac{|S|}{\Delta \tau} = O\left(\tfrac{|y(X,Y)|}{\lambda \Delta |X|}\right),
	\]
	 as needed. 

	For the third invariant, let $\OPT$ be (the indices of) an optimal increasing subsequence of $\blocklis(y,X,Y)$. Let $M_i \triangleq \max\{y(X_i \cap \OPT)\}$ be the maximal value of $y(\OPT)$ in $X_i$ (by convention, $M_i = \bot$ if $\OPT \cap X_i = \emptyset$). Let $U_i \triangleq (M_{j},M_i]$ for the largest $j<i$ such that $M_j$ is defined (similarly, let $U_i = \bot$ if $\OPT \cap X_i = \emptyset$). Last, define 
	\[
	R_i \triangleq \argmax_{[j,k] \subseteq U_i: j,k \in \cS_i}(k-j),
	\] i.e., the largest interval in $U_i$ using the values of $\cS_i$ only. We have the following bound on $\blocklis(y,X,Y)$:
	\[
	\blocklis(y,X,Y) = \sum_i \blocklis(y,X_i,U_i) \\
	\leq \sum_i \left(\blocklis(y,X_i, R_i) + \blocklis(y,X_i, U_i \setminus R_i) \right).
	\]
	Now, notice that by definition of $R_i$, $U_i \setminus R_i \cap \cS_i = \emptyset$. Notice also that $U_i \setminus R_i$ 
	consists of at most 2 intervals. We invoke \Cref{cl::sample_inv} (1) on each one of those intervals to obtain that $\sum_{w \in X_i} \1[y(w) \cap U_i \setminus R_i \neq \emptyset] \leq 4 \lambda r$, implying that $\blocklis(y,X_i, U_i \setminus R_i) \leq 4 \lambda r$ as well. We therefore have:
	\[
	\sum_i \blocklis(y,X_i, R_i) \geq \blocklis(y,X,Y) - 4 \lambda r \cdot \tau = \blocklis(y,X,Y) - 4 \lambda |X|
	\]

It remains to bound $\sum_i \blocklis(y,X_i, R_i)$. Here, we use \Cref{lm::int_cluster} (3), with $\psi(J) = \blocklis(y,X_i,J)$ and $I = R_i$, to obtain $\blocklis(y,X_i, R_i) \leq O(\log |\cS_i|) \cdot \blocklis(y,X_i, Y_i)$, for some $R_i \supseteq Y_i \in \cY_{i,*}$. Last, to avoid ``too large $\Delta$'', we simply do not count 
intervals 
with many sampled pivots. In particular, note that since $|\cS|\le 2k\tau/\lambda$ whp (by the immediate Chernoff bound), we have that at most $\lambda \tau$ blocks satisfy $|\cS_i|\ge 2k/\lambda^2$. 
Overall:
\begin{align*}
    \sum_{i: |S_i| \leq k/\lambda^2} \blocklis(y,X_i, Y_i) 
    &\geq \sum_{i: |\cS_i| \leq
2k/\lambda^2} \tfrac{\blocklis(y,X_i, R_i)}{O(\log |\cS_i|)}\\
    &\geq 
\tfrac{1}{O(\log (k/\lambda))}\left(\sum_{i \in [\tau]}\blocklis(y,X_i, R_i)- \lambda |X|\right) \\ 
    &=
\Omega \left(\tfrac{\blocklis(y,X,Y) - \lambda |X|}{\log (k/\lambda 
)}\right).
\end{align*}

We also note that since $\{(i,R_i)\}_i$ are monotone, so are  $Y_i$. This concludes the proof.
\end{proof}

It remains to prove \Cref{cl::sample_inv}.

\begin{proof}[Proof of~\Cref{cl::sample_inv}]
	Fix $Y' \in \cI$. Assume w.l.o.g $Y' \in \left\{(j,k) \mid \{j,k\} \in y(X_i) \cup \{-\infty, \infty\}\right\}$ (otherwise, consider the smallest interval in this set, which contains $Y'$). Notice that since there are at most $\poly(n)$ such intervals, it suffices to prove this holds w.h.p. for a single $Y'$ and the claim follows by the union bound.

	To prove the claim, we also define $\gamma \triangleq |Y' \cap \cS_i| + 2$ and use $\gamma$ to get our bounds.
	
	\paragraph{Concentration Analysis for $\gamma$.} First, we show concentration 
	for $\gamma$ using the quantities we care to bound. 
	We observe that $\gamma  - 2 \in \tfrac{|p^i_* \cap Y'|}{\zeta}+[0,1]$, and hence $\gamma \in \tfrac{|p^i_* \cap Y'|}{\zeta} + [2,3]$. Now, the quantity $|p^i_* \cap Y'|$ can be lower bounded by 
	$\Bin\left(\sum_{w \in X_i} \1[y(w) \cap Y' \neq \emptyset],\zeta\tfrac{\tau /\lambda}{|X|}\right) - \zeta$, which w.h.p. is at least:
\[
\tfrac{\tau \zeta}{\lambda} \cdot \tfrac{\sum_{w \in X_i} \1[y(w) \cap Y' \neq \emptyset]}{|X|} - 2\zeta = \zeta \cdot \left(\tfrac{\sum_{w \in X_i} \1[y(w) \cap Y' \neq \emptyset]}{\lambda r} - 2 \right).
\]
Hence $\gamma \geq \sum_{w \in X_i} \1[y(w) \cap Y' \neq \emptyset] \, / \, (\lambda r)$, 
and since $|Y' \cap \cS_i| + 2 = \gamma$, we obtain part (1) of the claim. 

To prove (2), we show that our sampling procedure guarantees that most blocks $w \in X_i$ will have bounded number of values 
in $Y'$.
	
	\paragraph{All blocks in $y(X_i)_{-\lambda/2,Y'}$ contain $O(\gamma)$ elements in $Y'$.}
	Let $T \triangleq \max_{w \in y(X_i)_{-\lambda/2,Y'}} \{|y(w) \cap Y'|\}$. We argue $T=O(\gamma)$. By construction, we have $\tfrac{\lambda r}{2}$ blocks in $y(X_i)$ with at least $T$ 
	integers in $Y'$. 
	So, the quantity $|p^i_* \cap Y'|$ is lower bounded by $T \cdot \Bin\left(\tfrac{\lambda r}{2},\zeta\tfrac{\tau/\lambda}{|X|}\right)$.
	Hence, by Chernoff bound, we have 
	\[
	\zeta \gamma \geq |p^i_* \cap Y'| \geq T \cdot \left(\tfrac{\zeta}{3} - O(\log n)\right) > \tfrac{\zeta T}{4}
	\]
	w.h.p. (for $\zeta$ with a high enough constant).
	This implies that $T \le 4\gamma$. 

Finally, observe that the quantity $|p^i_* \cap Y'|$ can also be lower bounded by viewing it as sampling blocks from 
$y(X_i)_{-\lambda/2,Y'}$, whose max-size (in $Y'$) is $T$. Hence, we have, by the Chernoff bound, w.h.p.:
\[
|p^i_* \cap Y'| \geq \zeta\tfrac{\tau/\lambda}{2|X|} \cdot |y_{-\lambda/2}(X_i,Y')| - \tfrac{1}{4} \, \zeta\, T.
\]
This implies that $ 
|y_{-\lambda/2}(X_i,Y')| / 2(\lambda r)\le \gamma+\gamma$,  
which further implies part (2) as well.

%
%
%
\end{proof}

\section{\genlis and \reslis Extensions: 
Proof of \Cref{lm::genlis_int_reduction,lm::reslis_int_reduction}}\label{sec:genlis_ext}


In this section we prove the reductions regarding the \genlis and \reslis extensions, namely \Cref{lm::genlis_int_reduction,lm::reslis_int_reduction}.

\subsection{\reslis extension}\label{sec::block_lis_extension}

For the \reslis problem, all we need to show is that one can ``ignore'' not too many exceptionally large blocks. The algorithm is the natural one:
\begin{enumerate}
	\item Using \Cref{cr::sampling_ds} and $T$, 
    sample blocks $S \subseteq X$
    with i.i.d. 
    probability $\tfrac{\zeta}{\lambda |X|}$, and let 
    $W \gets \{(i,j) \in S \times [k] \mid y_{i,j} \in Y\}$.
	\item Compute $\Upsilon \gets$ the $(1 - 0.6 \lambda)$-quantile 
	of $\{|\{(i,j) \in W\}|\}_{i \in S}$.
    \item Run $\cA$ with $\lambda' = 0.3\lambda$. Whenever $\cA$ tries to access a block $i$ using $T$, we check first if $|y_{i,*} \cap Y| > \Upsilon$, and if so, we treat the block as ``empty'' (say, by override).
   	\end{enumerate}

\begin{proof}[Proof of \Cref{lm::reslis_int_reduction}]
	First, we argue that each block in $y(X)_{-\lambda/2,Y}$ has at least $\Upsilon$ elements in $Y$. Indeed, that holds w.h.p.
	by the Chernoff bound. 
	Therefore, the instance we approximate is contained in $y_{-\lambda/2}(X,Y)$. It remains to show the additive error bounds. Here again, by the Chernoff bound, we ``omit'' at most $0.7\lambda n$ blocks, and hence the total additive error is at most $0.3\lambda n + 0.7\lambda n = \lambda n$ w.h.p. This concludes the proof.
	\ns{expand a bit...}
\end{proof}

\subsection{\genlis extensions}\label{sec::genlis_extension}

For the \genlis problem, we need to show two extensions:

\subsubsection{Sparse instances}
When the instance $g$ consists of $m \ll nk$ non-null elements\footnote{Note that unlike ``sparse genuine'' instances where we have few genuine elements, here we 
have few elements overall.}, we would like to improve the approximation and runtime bounds to be a function of  $m/n$ instead of $k$.

 To obtain the improved bound, we first exponentially discretize the blocks based on the number of non-null elements (denoted the {\em block size} $m_i$) and output the maximum over all instances. For each scale of block size $b \in E_2(k)$, we use the following extended version of \genlis:
 
 \begin{lemma}[$\estgenlis$ algorithm, sparse version]
\label{lm::gl_rec_ext}

    Suppose the algorithm $\estgenlis$ $(a(\lambda/k),\lambda n)$-approximates $\genlis$ for some function $a$ in time $t(n \cdot k)$. Then there exists an algorithm $\cA$ that $(O(a(\lambda n/\mathsf{d}k) 
    ),\lambda n)$-approximates $\genlis$ with $\mathsf{d}$ non-empty blocks. 
    The run-time is $t(\mathsf{d} k) + O(n)$.

\end{lemma}

We adapt algorithm \estgenlis as follows. Let $D$ be the set of non-empty blocks. We proceed by following the same algorithmic steps, using parameters $n' \triangleq \mathsf{d}$ 
and $\lambda' \triangleq \lambda n / \mathsf{d}$. 

\begin{proof}

It suffices to show that the algorithm output guarantees do not deteriorate from the case where the input is $g \in \N^{\mathsf{d} \times k}$ with parameter $\lambda'$ as above and does not include the empty blocks. Note that the additive error is $\lambda n = \lambda' \mathsf{d}$.

Indeed, for the dense estimator, we only count $i \in D$ which are subsampled, and each one is independently subsampled with the same probability as above.

For the sparse estimator, we note that the guarantees of \Cref{thm::reslis_main} are with respect to $\tfrac{|Y \cap y|}{|X|}$, and increasing $n$ can only help, since it gives us greater flexibility for choosing $\tau \in E_\beta$.
\ns{this is not super-formal right now... but I think okish}

The runtime is immediate, as the only overhead is skipping empty blocks.
\end{proof}
 
 
\subsubsection{\genlis over interval space}
We also need to overcome the obstacle that the space $\IC$ admits a partial order relation only, as opposed to $\N$, which admits a total order relation. Our ideal 
solution involves a mapping $\varphi: \IC \rightarrow \N$ that approximately preserves the overall $\LIS$ over all subsequences (with probability 1), and therefore also preserves the overall $\LIS$ over the genuine elements. We first note that such a mapping is applied a priori on the first coordinates of the input, and is not affected by the genuineness flags, nor does it require access to any genuineness samples.

While we are unable to exhibit 
a single mapping for the entire space, we 
show that one can 
discretize an arbitrary \genlis instance over intervals into $\log (k/\lambda)$ ``nicer'' instances $G_1,\ldots,G_{\log (k/\lambda)}$, which are mutually disjoint and whose union 
covers most of the original $\genlis$ intervals (the unused intervals are 
discarded, incurring 
some additive error). 
Each such ``nicer" instance is then mapped to a \genlis instance over integers, incurring a constant factor approximation 
over all its subsequences. The final output is 
the maximum \genlis over all integral instances $\{\varphi_i(G_{i,1}),G_{i,2}\}_i$. This costs 
us another $\log (k/\lambda)$-factor approximation.

In particular, we use the following claim:

\begin{claim}\label{cl::intervals_similar}
	For $\ell \in \N$, define $\cI_\ell$ to be the collection of all intervals in $\cI$ whose length is in 
	$[\ell, 2\ell)$. 
	Then, there exists a mapping $\varphi_\ell: \cI_\ell \rightarrow \N$, which can be computed in constant time, 
   such that:
    \begin{enumerate}
    	\item For $I,J \in \cI_\ell$: $I < J \Rightarrow \varphi_\ell(I)<\varphi_\ell(J)$; and,
    	\item For all increasing sequences $y \in \N^m$ and all $q \in \cI_\ell^m$ 
    	such that $\varphi_\ell(q)=y$, $q$ must have an increasing subsequence of length $m/3$.
   \end{enumerate}
\end{claim}

\begin{proof}
    We use the mapping $\varphi_\ell(I) \triangleq \lfloor \min(I)/\ell \rfloor$. 
    First we show $(1)$. We have 
    \[
    I<J \Rightarrow \min(I)\leq \max(I) - \ell + 1 \leq \min(J) - \ell,
    \]
    which implies that $\varphi_\ell(I)<\varphi_\ell(J)$.
    
   For (2), fix $y$ as above, and for $i \in \{0,1,2\}$, let $y^i$ be the sequence $y$ restricted to integer values $y_w$ satisfying 
   $y_w \pmod 3 = i$. 
   Let 
   $q^i$ be the corresponding subsequences 
   of $q$. Then each subsequence $q^i$ is 
   increasing. To see this, consider two intervals $q^i_j$ and $q^i_{j+1}$. Then $\max(q^i_j) < \min(q^i_j) + 2\ell$. On the other hand, $\min(q^i_{j+1}) > \min(q^i_j) + 2\ell$, as $y^i_{j+1} \geq y^i_j + 3$. So, $\max(q^i_j) <\min(q^i_{j+1})$.
   
   Since $|q^0| + |q^1| + |q^2| = m$, 
   at least one of them satisfies the length promise as well. 
\end{proof}



\newcommand{\rank}[0]{{\text{rank}}}

\begin{figure}[ht]
    \setstretch{1.2}
  \begin{algorithm}[H]
    \caption{\genlis 
   extension reduction}
		\DontPrintSemicolon
		\SetNoFillComment
		\KwIn{$g \in (\cI \times \zo)^{n \times k},\lambda\in [1/n,1]$, algorithm $\cA$ for \genlis over integers.}
		\KwOut{A number $\widehat{G} \in \left[0,n\right]$.
		}
		Let $y \gets g(*)_1$ be the first coordinates of $g$. \\ 
		
		Compute the set of integers $Q \gets \cup_{i,j} \{\min(y_{i,j}),\max(y_{i,j})\}$.\\
		
		Compute $a \in \cI^{n \times k}$ such that $a_{i,j} \gets [\rank_Q(\min(y_{i,j})),\rank_Q(\max(y_{i,j}))]$ for all $i,j$.\\
		
		\For{$(\ell,i) \in E_2(2k/\lambda) \times [n]$}
		{
			Compute $U_\ell \gets \{(i,j) \mid |a_{i,j}| 
			\in [\ell, 2\ell)\}$.\\
			Compute $b^{(\ell)} \gets \varphi_\ell(a(U_\ell)) \oplus g(U_\ell)_2$, 
			where $\varphi_\ell$ is the transformation of \Cref{cl::intervals_similar}.\\ 
			\For{$e \in E_2(k)$}
			{
				Compute $B \gets \{i \mid |b_i| 
			\in [e,2e)\}$.\\
				$L_{\ell,e} \gets \tfrac{1}{3} \cdot $ \{ the output of algorithm from \Cref{lm::gl_rec_ext} with inputs $b^{(\ell)}_B, \lambda/2$ \}. 
			}
		}
		\Return $\max_{\ell,e} L_{\ell,e}$. 
\end{algorithm}
\caption{Description of the \genlis 
extension algorithm (\Cref{lm::genlis_int_reduction}). Here, $\oplus$ denotes direct sum.}
\label{fig:ALG_genlis_intervals}
\end{figure}

\begin{proof}[Proof of 
\Cref{lm::genlis_int_reduction}]
    Let $\cA$ be such that $\cA(b^{(\ell)})$ is a $(\alpha, \lambda n/2)$-approximation for $\genlis(b^{(\ell)})$ in time $t$. We need to show the reduction algorithm provides a $(O(\log (1/\lambda) \alpha),\lambda n)$-approximation for $g$.
    
            For the upper bound, we use 
    \Cref{cl::intervals_similar}, item (1), to map each interval 
    to an integer preserving its strict order relation with all other intervals. This implies that we always output an upper bound of the original interval $\LIS$, 
    as we only use subsets of elements of the original sequence, and do not add elements.
    
    For the lower bound, first consider intervals 
    of length at least 
    $2k/\lambda$, which do not participate in any $b^{(\ell)}$. Any increasing subsequence of intervals contained in $[1,2nk]$ can have at most $\lambda n / 2$ intervals of such magnitude, hence by ignoring such intervals we ``omit" 
    at most  $\lambda n / 2$ $\LIS$ elements (additive error). Now, consider all 
    intervals of length at most $2k/\lambda$. Each such interval 
    participates in some instance $b^{(\ell)}$, and hence there exists an instance 
    such that the intervals participating in it $(\log(k/\lambda),\tfrac{\lambda}{2} n)$-approximates the original instance. Now, using 
    \Cref{cl::intervals_similar}, item (2), this instance is 
    mapped to an integer instance which $(O(\log(k/\lambda)),\tfrac{\lambda}{2} n)$-approximates the original instance. Therefore, the lower bound holds for at least one $b^{(\ell)}$. 
    
    Next, we argue that partitioning $b^{(\ell)}$ into $\log(k)$ instances of similar block size guarantees that at least one instance has $1/\log(k)$ fraction of the $\LIS$. Finally, 
    we use the guarantees of \Cref{lm::gl_rec_ext},
   together with
    the above bounds on $b^{(\ell)}_B$, to conclude that $\cA(b^{(\ell)}_B)$ outputs a $(\alpha \cdot O(\polylog(k/\lambda)), \lambda n)$-approximation of the original $\LIS$.    
    This concludes the proof.
    
\end{proof}

\begin{flushleft}
\bibliography{main}{}
\bibliographystyle{alpha}
\end{flushleft}

\appendix
\section{Proof of \Cref{thm::reslis_main}}
\label{apx:proofRecursion}


The proof is by recursive application of 
\Cref{thm::gl_rec_tree,thm::sl_rec_tree} for solving $\reslis$, up to approximation $\left(a\left(|X|,\tau,\tfrac{\lambda \cdot |X|}{|y(X, Y)|}\right),\lambda |X|\right)$, with access to a Precision-Tree $T_{1/c}$, for $c=c\left(|X|,\tau,\lambda,\tfrac{\lambda \cdot |X|}{|y(X, Y)|}\right)$.

We will establish the recurrences for quantities  $a(\cdot,\cdot,\cdot)$ and $c(\cdot,\cdot,\cdot,\cdot)$. 

\begin{claim}\label{cl::recurse}
    Fix $\eps < 1$ and $\xi,\theta,n \in \N$. Consider the following recursive formulas:
    \begin{align*}
        a(m,\tau,1/L) &\leq \xi \cdot \max_{t\in [1,O(L)]} \max\left\{\gamma,a(\tau,\tau^\eps,\tfrac{\gamma}{ t})\right\}\cdot a(m/\tau,\tau,\Omega(t/L)); \\
        c(m,\tau,1/D,1/L) &\le \theta \cdot \tau \cdot D+ \max_{r \in [1,D]}\left(\tO(r) + c(\tau,\tau^\eps,1/r,\gamma/t)\right)\cdot c(m/\tau,\tau,\Omega(r/D),\Omega(t/L)). 
    \end{align*}
    With the following base cases:
    \begin{align*}
        a(\cdot,\cdot,1) &= 1
        &a(m,\cdot,1/m) &= 1 
        &a(1,\cdot,\cdot) &= 1 \\
        c(\cdot,\cdot,\cdot,1) &= 1 
        &c(1,\cdot,\cdot,\cdot) &= 1 
        &c(m,\cdot,\cdot,1/m) &= m  
     \end{align*}
     
    
    Then, there exists a function $\gamma(t,n,m,\eps)$, such that substituting $\gamma = \gamma(t,n,m,\eps)$ in the recursive formulas above yield the following upper bounds, for $H \triangleq  \eps^{\Theta(1/\sqrt{\eps}\cdot \log 1/\eps)}$:
    
    \begin{align*}
        a(n,n^\eps,1/L) &= L^{\sqrt{\eps}}\cdot  \xi^{O(H^2)} \\ 
        c(n,n^\eps,1/D,1/L) &=  D \cdot m^{O(\eps \log_{(1/\eps)}(H))}\cdot \theta^{O(H)}. \\
    \end{align*}  
\end{claim}

The proof of \Cref{cl::recurse} is deferred to \Cref{sec::recurse}.

To prove \Cref{thm::reslis_main}, we deduce the recurrences on approximation and \prectree complexity showing they match the ones from \Cref{cl::recurse}, from which we directly obtain the final bounds on approximation. Then we analyze the runtime, where we use the bound established on $c(\cdot,\cdot,\cdot,\cdot)$ as well.

\begin{proof}[Proof of \Cref{thm::reslis_main}]

To deduce the recurrences, we first note that for base cases, we have $a(\cdot, \cdot, 1)=1$, $c(\cdot, \cdot, \cdot,1)=1$ as in this case $\lambda=1$ and additive approximation $|X|$ is trivial. Similarly, for $|X|=1$ (i.e., 1 element string), we have $a(1, \cdot, \cdot)=1$, $c(1, \cdot, \cdot,\cdot)=1$, as we can compute if $\LIS$ is 1 or 0 with access to a 1 element 
string. Similarly, we need full access for exact approximation, and hence: $a(m, \cdot, 1/m)=1$, $c(m, \cdot,\cdot,1/m)=m$. For the other cases, we show the following:

\paragraph{Approximation.}  
For convenience, let $m\triangleq |X|$, $L \triangleq  \tfrac{|y(X, Y)|}{\lambda \cdot |X|}\le k/\lambda$.
We apply
\Cref{thm::sl_rec_tree}. Define $t\triangleq k_g/\lambda_g$, in which case $\tfrac{\lambda_s|X|}{|y(X)\cap Y|}=\Omega(t/L)$. Note that if $t>\Omega(L)$, then, we get (base case) $a(|X|,\cdot, 1)=1$, and hence $a(m,\tau,1/L)\le(\log k/\lambda)^{O(1)}\cdot \max_{t\le 1/\lambda} a_g(\tau,1/t)$ . Thus, including when $t$ is smaller, we have that:
	$$a(m,\tau,1/L) \le (\log k/\lambda)^{O(1)}\cdot \max_{t\in [1,O(L)]} a_g(\tau,1/t)\cdot a(m/\tau,\tau,\Omega(t/L)).$$
	
	Plugging in $a_g$ from 
	\Cref{thm::gl_rec_tree}, (with $\tau_s=n^\eps$), 
	we get, for our choice of $\gamma$:
	
	$$a(m,\tau,1/L) \le (\log k/\lambda)^{O(1)} \cdot \max_{t\in [1,O(L)]} 
	\max\left\{\gamma,a(\tau,\tau^{\eps},\tfrac{\gamma}{ t})\right\}\cdot a(m/\tau,\tau,\Omega(t/L)).$$

\paragraph{\prectree Complexity.} 
We now deduce the complexity bound $c$. Note that the recursion for $c$ will be using same parameters as above. We need to keep track of new parameters, in particular let $D=1/\lambda$. We have, where $\lambda_s,\lambda_g,k_g$ are as in approximation bound above, with $\lambda_s\lambda_g=\Omega(k_g\lambda)$:
$$
c(m,\tau,1/D,1/L) \le \polylog(n) \cdot  \tau D+ (c_g(\tau,\lambda_g,k_g)\cdot c(m/\tau,\tau,\Theta(\lambda_s/k_g),\Omega(t/L)) 
$$

Here we set $r=1/\lambda_g\in[1,D]$ and hence:
$$
c(m,\tau,1/D,1/L) \le \polylog(n) \cdot \tau L+ c_g(\tau,1/r,k_g)\cdot c(m/\tau,\tau,\Omega(r/D),\Omega(t/L)).
$$

Plugging in $c_g$ from \Cref{thm::gl_rec_tree} (again, with $\tau_s=n^\eps)$, we get:
$$
c(m,\tau,1/D,1/L)\le \polylog(n) \cdot \tau D +\max_{r \in [1,D]} \left(\tO(r) + c(\tau,\tau^{\eps},1/r,\gamma/t)\right)\cdot c(m/\tau,\tau,\Omega(r/D),\Omega(t/L)) ,
$$

We now apply \Cref{cl::recurse}, with $\theta = \polylog(n)$ and $\xi = \polylog (k/\lambda)$ to obtain the bounds we need.


\paragraph{Runtime complexity.}
The recurrence relation $t_s(n,n^\eps) = \kappa \cdot t_s(n^{1-\eps},n^{\eps})$, where $\kappa = \tO\left(\log^{O(\log_\beta(\tau))}(\beta)\right)$, with base case $t_s(n^\eps,n^\eps)=\kappa$, can be expanded to obtain $T(n,n^\eps) = \kappa^{1/\eps}$.


Therefore, the total expected runtime is $\kappa^{1/\eps} \cdot c(m,\tau,1/D,1/L) \cdot \tfrac{L}{D} = L \cdot m^{O(\sqrt{\eps}\log 1/\eps)}\cdot (\log n)^{2^{O(1/\sqrt{\eps}\cdot\log^2 1/\eps)}} \cdot \tO\left(\log^{O(\log_\beta(n^{\epsilon}))}(\beta)\right)^{1/\epsilon} = \log^{O(\log_\beta(n))}(\beta) \cdot \log(n)^{O(1/\epsilon)}$. The preprocessing of the \prectree takes time $c\cdot \tO\left(\log^{O(\log_\beta(n^{\epsilon}))}(\beta)\right)$, which is a lower 
order term.

\end{proof}

\subsection{Recursion Analysis}\label{sec::recurse}
We will prove each recursion separately. Let $b = n^{1/H}$. Our choice of $\gamma = t^{\sqrt{\eps}}$, except for $\tau=b$, then we set $\gamma = t$ and get the base cases        $a(m,b,1/L) \leq \xi \cdot \max_{t\in [1,O(L)]} t \cdot a(m/b,b,\Omega(t/L))$ and 
        $c(m,b,1/D,1/L) \leq \theta \cdot b \cdot D+ \max_{r \in [1,D]}\left(\tO(r) +  c(m/b,b,\Omega(r/D),\Omega(t/L))\right)$.
\begin{proof}[Proof of \Cref{cl::recurse}, function $a$]
Setting $\tau = m^\eps$, the recursion formula can be written as:
$$
a(m,m^\eps,1/L) \le \xi^{\Upsilon}\cdot \max_{\stackrel{t_1,\ldots,t_{1/\epsilon}\ge 1}{t_1 \cdot t_2 \cdot \ldots \cdot t_{1/\epsilon} \le 2^{O(\Upsilon)} L}} \prod_i\max\{\gamma_i,a(m^{\eps},m^{\epsilon^2},\gamma_i/t_i)\},
$$
where $\Upsilon=1/\eps$ and $\gamma_i = \gamma(t_i)$. Let $f=\Upsilon \cdot \log (\xi)$.
Setting $q_h(l)=\log a(m^{\eps^h},m^{\eps^{h+1}},1/2^l)$, we obtain from above,  and substituting $\gamma_i$ with $\log \gamma_i = \sqrt{\eps}t_i$: 
$$
q_h(l)\le f + \max_{\stackrel{t_1,\ldots,t_{\Upsilon}\ge 0}{t_1 + t_2 + \ldots + t_{\Upsilon} \le l+O(\Upsilon)}} \sum_i\max\{\gamma_i,q_{h+1}(t_i-\gamma_i)\}.
$$
	
Now let $H' = \log_{\Upsilon}(H)$ and hence $b = n^{1/H} = n^{\epsilon^{H'}}$. 
Then $q_{H'}(l)\le f+l+O(\Upsilon)$ by the base case. 

Now for $h \leq H'$ we prove, by induction, that
$$
q_h(l)\le (f+O(\Upsilon))\cdot \Upsilon^{2(H'-h)}+(1-\sqrt{\eps})^{H'-h}\cdot l.
$$
Indeed, we first check the base case for $h=H'$: $q_{H'}(l)\le f+l+\Upsilon$.
Now note that
$q_h(t_i-\gamma(t_i))\le f\cdot \Upsilon^{H'-h}+(1-\sqrt{\eps})^{H'-h+1} \cdot t_i$. Hence, we can verify the inductive step for $h<H'$:
\begin{align*}
q_h(l)
\le& 
f + \max_{\stackrel{t_1,\ldots,t_{\Upsilon}\ge 0}{t_1 + t_2 + \ldots + t_{\Upsilon} \le l+O(\Upsilon)}} \sum_i\max\{\gamma(t_i),q_{h+1}(t_i-\gamma(t_i))\}.
\\ \le &
f + \max_{\stackrel{t_1,\ldots,t_{\Upsilon}\ge 0}{t_1 + t_2 + \ldots + t_{\Upsilon} \le l+O(\Upsilon)}} \sum_i (f+O(\Upsilon))\cdot \Upsilon^{2(H'-h-1)}+(1-\sqrt{\eps})^{H'-h} \cdot t_i
\\ \le &
f+(f+O(\Upsilon))\cdot \Upsilon^{2(H'-h)-1}+(1-\sqrt{\eps})^{H'-h}(l+O(\Upsilon))
\\ \le &
(f+O(\Upsilon))\cdot \Upsilon^{2(H'-h)}+(1-\sqrt{\eps})^{H'-h} \cdot l
\end{align*}

where we omitted $\gamma(t_i)$ from the max in the 2nd inequality since $\sqrt{\eps} t_i$ is smaller than the inductive hypothesis on $q_{h+1}(t_i)$ (using the particular choice on the value of $H$).

In particular, this bound implies $q_0(l)=O(f)\cdot \Upsilon^{2H'}+\sqrt{\eps}l = f\cdot O(H^2\cdot \Upsilon) +\sqrt{\eps}l$. 

Substituting back, we obtain that $a(m,m^\eps, 1/L)\le 2^{O(H\Upsilon) \cdot f+\sqrt{\eps} \log (L)}\le L^{\sqrt{\eps}}\cdot \xi^{O(H^2)}$.
\end{proof}

\begin{proof}[Proof of \Cref{cl::recurse}, (function $c$)]

Similar to the previous proof, setting $\tau=m^\eps$ with $t_i, \Upsilon$ from before too:
$$
c(m,m^\eps,1/D,1/L) \le (\log(n))^{O(\Upsilon)} \cdot\left(\sum_{i=1}^{\Upsilon}\tfrac{m^\eps D/r_1..r_{i-1}}{\log n} \prod_{j<i} c(\tau,\tau^\eps,1/r_j,\gamma(t_j)/t_j)\right),
$$
where the $r_i's$ satisfy $r_1,\ldots r_\Upsilon\ge 1 $ and $r_1r_2\ldots r_\Upsilon= D$. We further can simplify by dividing by $D$:
$$
\tfrac{1}{D}c(m,m^\eps,1/D,1/L) \le (\log(n))^{O(\Upsilon)} \cdot\tfrac{m^\eps}{\log n}\cdot\left(\sum_{i=1}^{\Upsilon} \prod_{j<i} \tfrac{1}{r_j}c(\tau,\tau^\eps,1/r_j,\gamma(t_j)/t_j)\right).
$$
As before, let $c_h(D)=\tfrac{1}{D}c(m^{\eps ^h},m^{\eps^{h+1}},1/D,1/L)$, for which we get:
$$
c_h(D)=(\log n)^{O(\Upsilon)}\cdot m^{\epsilon^{h+1}}\cdot  \max_{i} \prod_{j<i} c_{h+1}(r_j),
$$
with the base case $c_{H'}(D)=(\log n)^{O(1)}$ (when we are taking the maximal $\gamma$). Hence we obtain that
$$
c_0(D)\le (\log n)^{\Upsilon^{O(H')}}\cdot m^{\eps\cdot H'}.
$$
Plugging back into $c$ we obtain that:
$$
c(m,m^\eps,1/D,1/L) \le D\cdot (\log n)^{\Upsilon^{O(H')}}\cdot m^{\eps\cdot H'}.
$$

\end{proof}



\end{document}